\documentclass[twocolumn, pra,footinbib, superscriptaddress  ]{revtex4-2}

\usepackage{times}
\usepackage{graphicx}
\usepackage{amsfonts}
\usepackage{amssymb}
\usepackage{amsthm}
\usepackage{array}
\usepackage{amsmath}
\usepackage{verbatim} 
\usepackage{hyperref}
\usepackage{color}
\usepackage{bbold}
\usepackage{epstopdf}
\usepackage{mathtools}
\usepackage{enumerate}
\usepackage{subcaption}
\usepackage{physics}
\usepackage{subcaption}
\usepackage{hyphenat}
\usepackage{multirow}
\usepackage{braket}

\newtheorem{proposition}{Proposition}
\newtheorem*{proposition*}{Proposition}

\newtheorem{theorem}{Theorem}
\newtheorem*{theorem*}{Theorem}
\newtheorem{corollary}{Corollary}[theorem]
\newtheorem*{corollary*}{Corollary}

\theoremstyle{definition}

\begin{document}

\title{\large Achieving Heisenberg limit under noisy conditions with quantum Zeno dynamics and dynamical decoupling}

\author{Ke Zeng}
\affiliation{ Institute of Fundamental and Frontier Sciences,
University of Electronic Science and Technology of China, Chengdu 610051, China}
\affiliation{Key Laboratory of Quantum Physics and Photonic Quantum Information, Ministry of Education, University of Electronic Science and Technology of China, Chengdu 611731, China}

\author{Bakmou Lahcen}
\affiliation{ Institute of Fundamental and Frontier Sciences,
University of Electronic Science and Technology of China, Chengdu 610051, China}
\affiliation{Key Laboratory of Quantum Physics and Photonic Quantum Information, Ministry of Education, University of Electronic Science and Technology of China, Chengdu 611731, China}

\author{Yu Jiang}
\affiliation{ Institute of Fundamental and Frontier Sciences,
University of Electronic Science and Technology of China, Chengdu 610051, China}
\affiliation{Key Laboratory of Quantum Physics and Photonic Quantum Information, Ministry of Education, University of Electronic Science and Technology of China, Chengdu 611731, China}

\author{Kok Chuan Tan}
\email{bbtankc@gmail.com}
\affiliation{ Institute of Fundamental and Frontier Sciences,
University of Electronic Science and Technology of China, Chengdu 610051, China}
\affiliation{Key Laboratory of Quantum Physics and Photonic Quantum Information, Ministry of Education, University of Electronic Science and Technology of China, Chengdu 611731, China}
\affiliation{School of Physics, University of Electronic Science and Technology of China, Chengdu 610054, China}

\date{\today}

\begin{abstract}
Quantum Zeno dynamics (QZD) and dynamical decoupling (DD) are useful tools that enable the effective suppression of noise in quantum systems. We consider the problem of when (i) noise can be suppressed and (ii) Heisenberg limit (HL) can be achieved in quantum metrology, and prove necessary and sufficient conditions for when QZD and DD are useful for achieving these two goals. We further show this condition can hold even when the Hamiltonian-not-in-Lindblad-span (HNLS) condition fails, allowing QZD/DD to recover Heisenberg scaling beyond what is certified by the usual HNLS-based quantum error correction (QEC) criterion. Finally, we demonstrate that the combination of both techniques can allow individually imperfect QZD and DD strategies to saturate HL.
\end{abstract}
\maketitle

\section{Introduction}
The goal of quantum metrology is to attain the highest level of precision allowed by the laws of quantum mechanics~\cite{caves1981quantum,giovannetti2004quantum,pezze2018quantum}. A metrology experiment typically consists of (i) preparation of a probe state, (ii) interaction with the probe state that depends on a physical parameter $\omega$, followed by (iii) a suitable measurement on the final state to estimate the value of $\omega$. A major obstacle to this is environmental noise~\cite{breuer2002theory}, which imposes strict limitations on the achievable precision. In such cases, the best achievable precision typically scales with $\sim 1/t$, where $t$ is the interaction time with the encoding system. This is called the standard quantum limit (SQL). In contrast, the fundamental limit allowed by quantum mechanics is  $\sim 1/t^2$, which is called the Heisenberg Limit (HL). Equivalently, one may also choose to distribute the encoding of $\omega$ across multiple correlated probes, which leads to similar $\sim 1/N$ (SQL) or $\sim 1/N^2$ (HL) scalings, where $N$ is the number of probes. 

Quantum Zeno dynamics (QZD) and dynamical decoupling (DD) are powerful tools that enable the effective suppression of noise in quantum systems. QZD is based on the well known quantum Zeno effect, where the continuous observation of a quantum system can effectively freeze the system in its initial state~\cite{von2013mathematische,sakurai2020modern,presilla1996measurement,nakazato1996temporal,panov1996quantum, koshino2005quantum}. Frequent observation of an unstable quantum system can inhibit its natural time evolution or decay~\cite{osti_4554238,osti_4825360,misra1977zeno,chiu1977time}.  Quantum systems can also be constrained to evolve within some subspace of a larger Hilbert space, which is then called quantum Zeno dynamics (QZD)~\cite{facchi2008quantum,smerzi2012zeno, lerner2018quantum,schafer2014experimental}.  In relation to quantum metrology, it has been shown that for non-Markovian noise $\sim 1/N^{3/4}$ scaling is achievable in the Zeno limit, i.e., when measurement times are sufficiently short~\cite{chin2012quantum,matsuzaki2011magnetic}. Recently, QZD was applied to achieve HL in the presence of amplitude damping noise~\cite{Liu2025}.

DD is a method of removing unwanted couplings by applying strong controls to the system~\cite{viola1998dynamical,viola2002quantum,cappellaro2006principles,uhrig2007keeping}, and is closely related to QZD. For example, ``bang-bang" decoupling is a DD technique that has been shown to be formally equivalent to QZD~\cite{facchi2004unification, hahn2022unification} via von Neumann’s ergodic theorem~\cite{reed1980methods}. Furthermore, QZD can be achieved by applying repeated kicks to the system, similar to many DD protocols~\cite{burgarth2020quantum,kato1950adiabatic,burgarth2019generalized,burgarth2021eternal}. Current literature suggests that DD and QZD are closely related, complementary techniques that suppress environmental noise.

These observations raise a fundamental question: under what general microscopic conditions can QZD or DD suppress system–environment coupling without simultaneously erasing the metrological signal?

Our methodology assumes a general time-independent system–environment interaction, with no other assumptions about the nature of the noise. We derive necessary and sufficient conditions for QZD in the Zeno limit and for leading-order DD in the fast-control limit to suppress environmental coupling while preserving nontrivial parameter encoding. Both are governed by the requirement that the signal Hamiltonian lie outside the span generated by the identity and the system-side interaction operators, which we refer to as the Hamiltonian-not-in-interaction-span (HNIS) condition. In contrast to the Hamiltonian-not-in-Lindblad-span (HNLS) condition developed for QEC-assisted metrology under Markovian dynamics~\cite{zhou2018achieving,Mann2025}, HNIS is formulated directly at the microscopic Hamiltonian level without assuming a Markovian Lindblad description~\cite{manzano2020short}. In the Markovian regime, HNIS can remain satisfied even when HNLS fails for the corresponding Lindblad model. Finally, we derive a sufficient condition for hybrid QZD–DD protection and show that combining the two techniques can recover Heisenberg scaling when either strategy alone is imperfect.

\section{Results}
\subsection{QZD and DD for quantum metrology}

Consider the Hilbert space $\mathcal{H} = \mathcal{H}_S \otimes \mathcal{H}_E$, where $\mathcal{H}_S$ and $\mathcal{H}_E$ are the Hilbert spaces of the system and environment respectively. The total Hamiltonian can be written as
\begin{equation} \label{eqn::H_tot}
   H_{\rm tot}= H_S\otimes \mathbb{1}_E+\mathbb{1}_S\otimes H_E+H_{SE},
\end{equation} where $H_S$ and $H_E$ are local Hamiltonians acting on subsystem $S$ and $E$,  while $H_{SE}$ describes the interaction between the system and environment. According to the Stinespring Dilation Theorem~\cite{Watrous2018}, any quantum noise process can be represented as a unitary evolution $U_{\rm tot} = \exp(-iH_{\rm tot} t)$ on the state $\rho_S \otimes |0\rangle  _E \langle 0|$ followed by a partial trace over the environment. Assume that $H_{\rm tot}$ is time independent and the system Hamiltonian $H_S = \omega G $ is parametrized by some real physical parameter $\omega$.  The goal is to estimate $\omega$ with the highest possible level of precision. It is known that for the estimation of a single parameter $\omega$, the ultimate achievable precision limit is given by the quantum Cramér-Rao bound~\cite{Rao1992}
\begin{equation}
\delta ^2 \omega  \ge \frac{1}{\nu F},
\end{equation} where $\nu$ is the number of repeated experiments and $F$  is the quantum Fisher information (QFI) defined as $F= \operatorname{Tr}[\rho_S L_\omega^2]$
where $\rho_S$ is the probe state and $L_\omega$ is the symmetric logarithmic derivative satisfying $\frac{d\rho_S}{d\omega} = \frac{1}{2}(\rho_S L_\omega + L_\omega \rho_S)$. We say that HL is achieved if $F$ scales with $\sim t^2$.

Let $\Pi_p = \Pi_p^2$ be a projector onto a subspace $\mathcal{H}_p$ of $\mathcal{H}$ and $U(t) = \exp(-iH_{\rm tot} t)$ be the unitary governing the time evolution of the total system. It can be shown that repeated projections onto $\mathcal{H}_p$ lead to an effective unitary evolution $U_Z$ acting on $\mathcal{H}_p$ in the quantum Zeno limit~\cite{facchi2008quantum}: 
\begin{equation} \label{eqn::QZD}
U_Z(t) = \lim_{n\rightarrow \infty} [\Pi_p U(t/n)\Pi_p]^n = \Pi_p \exp (-i H_Z t),
\end{equation} where $H_Z = \Pi_p H_{\rm tot} \Pi_p$ is the effective Zeno Hamiltonian. For any initial state $\rho_{SE}$ belonging to $\mathcal{H}_p$,  the projector $\Pi_p$ in front can be dropped and the evolution is just $\exp (-i H_Z t)$. This is a direct consequence of the fact that $\Pi_p \ket{\psi}_{SE} = \ket{\psi}_{SE}$ for any $\ket{\psi}_{SE} \in \mathcal{H}_p$. This effective dynamics is called QZD.

The primary goal of DD is to add an additional time dependent control term $H_C(t)$ acting on $\mathcal{H}_S$, transforming the total Hamiltonian: \begin{equation}
    H_{\rm tot}\left(t\right)= H_S\left(\omega\right)\otimes \mathbb{1}_E+ H_C\left(t\right)\otimes \mathbb{1}_E+\mathbb{1}_S\otimes H_E+H_{SE}.
\end{equation} Let ${U_C}\left( t \right) = {\mathcal{T}}\exp\left( { - i\int\limits_0^t {ds{H_C}\left( s \right)} } \right)$, where $\mathcal{T}$ denotes the time-ordering operator, then in the control frame the time evolution of the global state is described by the unitary operator: \begin{equation}{\tilde  U_{\rm tot}}\left( \omega, t \right) = {\mathcal{T}}\exp\left( { - i\int\limits_0^t {ds{\tilde  H_{\rm tot}}\left( s \right)} } \right),\end{equation} where ${\tilde H_{\rm tot}}\left( t \right) = \left( U_C^{\dagger}\left( t \right)\otimes \mathbb{1}_E\right){H_{\rm tot}}\left( U_C\left( t \right)\otimes \mathbb{1}_E\right)$.  
In this article, we will work within the ideal bang--bang control regime, in which the control pulses are instantaneous and have unbounded strength, while implementing finite unitary transformations on the system. Let $t=nT$, and let $U_C(kT)$ denote the toggling-frame control operator during the $k$th time interval of duration $T$. 
The first term of the Magnus expansion gives
\begin{equation}\label{eqn::first_order_Htot}
\overline{H}_n^{(0)}=H_{S}^{\text{eff}}\otimes\mathbb{1}_E+\mathbb{1}_S\otimes H_E+H_{SE}^{\text{eff}},
\end{equation}
where $H_S^{\text{eff}} = \frac{1}{n}\sum\nolimits_{k = 0}^{n - 1} {U_C^{\dagger}\left( {kT} \right){H_S}\left( \omega  \right)U_C{{\left( {kT} \right)}}}$, and $H_{SE}^{\text{eff}} = \frac{1}{n}\sum\nolimits_{k = 0}^{n - 1} {\left( {U_C^{\dagger}\left( {kT} \right) \otimes {\mathbb{1}_E}} \right){H_{SE}}\left( {U_C{{\left( {kT} \right)} } \otimes {\mathbb{1}_E}} \right)}$.

Leading-order decoupling is achieved when the system $H_{SE}^{\text{eff}}= c \mathbb{1}_S \otimes J_E$ for some choice of $H_C(t)$, constant $c$ and operator $J_E$.
Accordingly, all DD statements in the subsequent sections are within the context of the bang--bang framework. In this framework, the higher-order Magnus contributions vanish in the fast-control limit, so that the evolution generated by the first Magnus term becomes asymptotically exact. A derivation of this fact is provided in  Methods~\ref{Method:magnus_of_DD} for the convenience of the reader. Throughout this work, achieving HL via QZD/DD refers to protection protocols that suppress the effective system--environment interaction while preserving nontrivial parameter encoding.

$H_{\rm tot}$ describes a general time independent system environment interaction. As a special case, we will also consider Markovian interactions, which are described by the Lindblad master equation~\cite{manzano2020short, Lindblad76} 
\begin{equation} \label{eqn::Lindblad}
   \dv{\rho_S}{t} = -i[H_S(\omega), \rho_S] + \sum_{i=1}^r\left (  L_i \rho_S L_i^\dagger - \frac{1}{2}\acomm{L_i^\dagger L_i}{\rho_S} \right),
\end{equation} where $L_i$ are called Lindblad jump operators.

Finally, we will use the Hilbert-Schmidt inner product $\langle A,B \rangle \coloneqq \operatorname{Tr}(A^\dagger B)$. Orthogonality of matrices and operator subspaces is defined with respect to this inner product.

\subsection{Necessary and sufficient conditions for QZD} \label{sec::QZDmain}

Consider $H_{\rm tot}$ from Eq.~\ref{eqn::H_tot}. The interaction part of the Hamiltonian can always be written as a sum over products $H_{SE} = \sum_i S_i \otimes E_i$,  where $S_i , E_i$ have non-zero support. Assume also that $\{ E_i \}$ forms a linearly independent set, which can always be achieved by taking $\{ E_i \}$ to be any basis in operator space. Under these light assumptions, we would like to know when QZD is useful in quantum metrology.

Our first result solves this problem and provides a linear algebraic condition that is simple to check, assuming that the system environment interaction $H_{SE}$ is known. 
 For the remainder of this article, we define the interaction span as the operator subspace $\mathcal{I}_{SE}\coloneqq\operatorname{span}\{\mathbb{1},S_i\}$. We say that the signal Hamiltonian $H_S$ satisfies the HNIS condition when $H_S\notin\mathcal{I}_{SE}$.
This result is stated as the following theorem:

\begin{theorem} \label{thm::QZD}
Consider the Hamiltonian $H_{\rm tot}$ where $H_S = \omega G$ and $H_{SE} = \sum_i S_i \otimes E_i$, where $S_i , E_i$ have non-zero support, and $\{E_i\}$ is linearly independent. Then, assuming we have access to a noiseless ancilla, the environment can be decoupled from the system and the parameter $\omega$ can be estimated with HL precision via QZD iff HNIS the condition holds. The projector $\Pi_p$ corresponding to the QZD subspace acts only on the system of interest $S$ and the ancilla.
\end{theorem}

Theorem~\ref{thm::QZD} is the first main result in this article, and provides a linear algebraic condition that can be systematically checked to determine if HL is reachable using a QZD strategy. The requirements for doing this is that the generator $G$ in $H_S=\omega G$, as well as the interaction Hamiltonian $H_{SE}$ are known completely. The decomposition $H_{SE} = \sum_i S_i \otimes E_i$ can always be performed by choosing any set of linearly independent operators acting on the environment subsystem $E$. The HNIS condition is superficially quite similar to the HNLS condition for QEC in Ref.~\cite{zhou2018achieving}, however, we note that there are regimes where one condition can be satisfied but not the other, making QZD and QEC complementary techniques for achieving HL. The differences between the HNIS and HNLS conditions are discussed in greater detail under Section~\ref{sec::HNISvsHNLS}.

\subsection{Necessary and Sufficient Conditions for DD} \label{sec:: DDmain}
Previously, we proved exact conditions where QZD is useful for quantum metrology. We will now consider the conditions where DD is useful under similar assumptions. For QZD, we needed to find projector $\Pi_p$ where the effective dynamics suppresses the noise without suppressing the signal encoding of the parameter $\omega$. For DD, the goal is to find control Hamiltonian $H_C(t)$ acting on subsystem $S$ or equivalently some set of control unitary operators $\{ U(kT) \}_{k=0}^{n-1}$  such that (i) $H_{SE}^{\text{eff}} = c \mathbb{1}_S \otimes J_E$  and (ii) $H_S^{\text{eff}} \not\propto \mathbb{1}_S$ are satisfied (see Eq.~\ref{eqn::first_order_Htot}). Condition (i) decouples the system-environment interaction, thus suppressing noise, while condition (ii) ensures nontrivial unitary encoding of parameter $\omega$ so that HL can be achieved.  Our second main result is stated as the following theorem:

\begin{theorem}
\label{thm::DD}
Consider the Hamiltonian $H_{\rm tot}$ where $H_S = \omega G$ and $H_{SE} = \sum_i S_i \otimes E_i$, where $S_i, E_i$ have non-zero support, and $\{E_i\}$ is an orthonormal set. Then the environment can be decoupled from the system and the parameter $\omega$ can be estimated with HL precision via DD  iff the HNIS condition holds. 
\end{theorem}  

As part of the proof of sufficiency in Theorem~\ref{thm::DD}, a mixed unitary channel $\Phi_{K,\beta}$ that achieves HL whenever the condition is satisfied is constructed. It can be shown that this channel is always optimal for single qubit systems, and optimal for a class of multi-qubit Hamiltonians. A channel $\Phi$ is said to be optimal in the sense that if one decomposes $H_S$ into its parallel and orthogonal components such that $H_S^\parallel \in \operatorname{span}\{ \mathbb{1}, S_i\}$, $ H_S^\perp \perp   \operatorname{span}\{ \mathbb{1}, S_i\} $, and $H_S = H_S^\parallel + H_S^\perp$, then the channel $\Phi$ satisfies $\Phi(H_S^\perp) = H_S^\perp$. This is the best that any mixed unitary channel can do. To see this, suppose the system Hamiltonian only has the parallel component such that $H_S = H_S^\parallel$. HL can never be achieved in this situation according to Theorem~\ref{thm::DD}, so parallel component $H_S^\parallel$ does not contribute to achieving HL, so the best one can do is to perfectly preserve the orthogonal component, i.e. $\Phi(H_S^\perp) = H_S^\perp$. The construction and discussion of the optimality of the channel $\Phi_{K,\beta}$, together with the proof of Theorem~\ref{thm::DD}, is presented in Methods~\ref{Method:HNIS_with_DD}.

QZD realizes a general code-space compression, whereas DD is restricted to mixed-unitary averaging on the system. Remarkably, despite the fact that the effective Hamiltonians $H_Z$ in QZD, and $H^{\text{eff}}_S$ in DD are defined differently, an essentially identical HNIS condition applies to both QZD and DD. There are, however, some physical differences in terms of the actual implementation of both approaches. Theorem~\ref{thm::QZD} assumes access to a noiseless ancilla, which is similar to the QEC approach~\cite{zhou2018achieving} where the HNLS condition applies. In contrast, DD does not require a noiseless ancilla. However, since we are working within the bang--bang regime, it does assume the ability to perform instant and arbitrarily strong unitary operations on the system of interest. Again, this is an assumption that overlaps with but is not completely identical to the assumptions made in the QEC approach. The actual requirement for achieving HL using any of the QZD, DD and QEC approaches may therefore differ drastically, depending on the technical constraints of the experiment.

\begin{figure*}[t]
    \centering
    \begin{minipage}[t]{0.95\textwidth}
        \includegraphics[width=\textwidth]{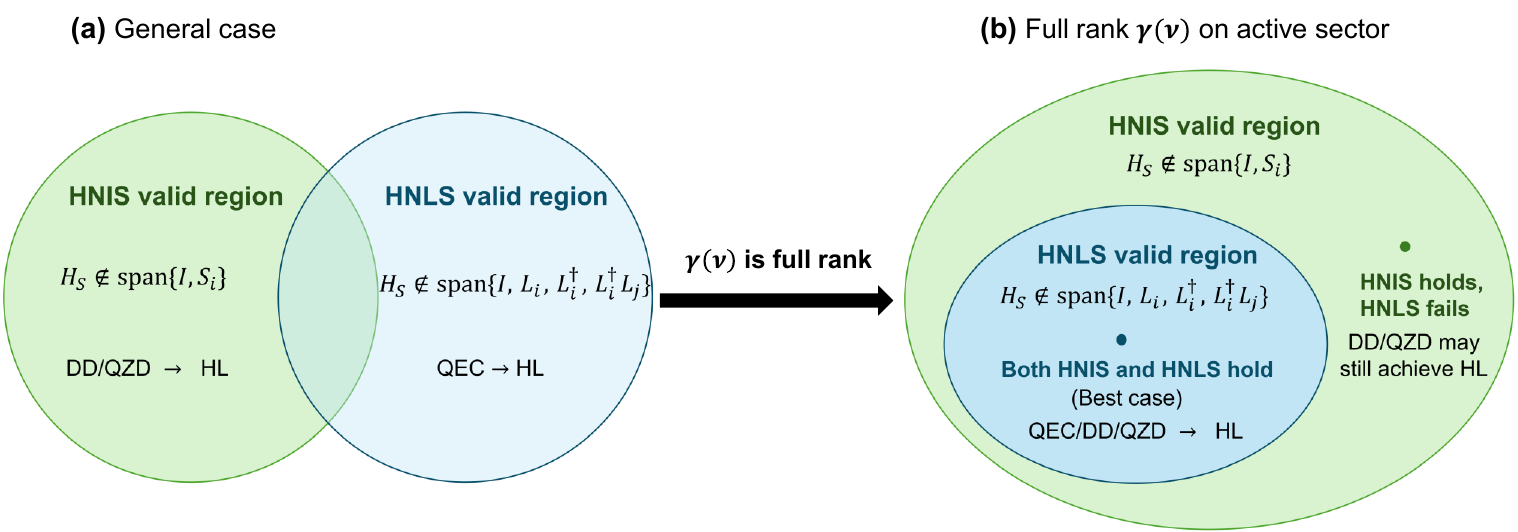}
        \caption{Relation between HNIS and HNLS.
        \textbf{(a)} In general, the HNIS- and HNLS-valid regions need not be nested.
        \textbf{(b)} If the bath correlation matrix $\gamma(\nu)$ is full rank for every Bohr-frequency sector, Proposition~\ref{thm::Markov} shows that HNLS implies HNIS, while the converse does not hold in general, as shown by Proposition~\ref{prop::HNLS}.}
        \label{fig:HNIS}
    \end{minipage}
\end{figure*}

\subsection{Relation between HNIS and HNLS} \label{sec::HNISvsHNLS}
Recall that Markovian time evolution is described by the Lindblad master equation (see Eq.~\ref{eqn::Lindblad}). Throughout this work, QEC refers to the standard framework underlying the HNLS theorem, where the time evolution of the probe is described by a fixed Markovian Lindblad generator, and it is assumed that noiseless ancillas and arbitrarily fast and accurate control and recovery operations are available. Under such conditions, the achievability of HL is well understood due to an important result from Ref.~\cite{zhou2018achieving}, which states that HL is achievable via QEC if and only if $H_S \notin \operatorname{span}\{ \mathbb{1}, L_i, L^\dagger_i, L^\dagger_i L_j\}$, where $L_i$ are the Lindblad jump operators describing the Markovian noise. This is known as the HNLS condition. A sufficient condition was recently derived for hidden Markov models~\cite{Mann2025}. 

HNIS and HNLS are formulated at different descriptive levels. HNIS is defined directly from the microscopic system--environment interaction, whereas HNLS is defined from the jump operators of a fixed Markovian Lindblad generator. In general, neither condition implies the other. However, when the latter is derived from the former under the Born--Markov and secular approximations, and assuming that the bath correlation matrix $\gamma(\nu)$ is full rank on every active Bohr-frequency sector, HNLS implies HNIS. The converse does not hold in general. This relation is illustrated in Fig.~\ref{fig:HNIS} and stated in the Proposition below.

\begin{proposition} \label{thm::Markov}
     Given any microscopic description of the system-environment interaction $H_{\rm tot}$, let $\{{ L_i(\nu)}\}$ be the corresponding Lindblad jump operators under the Born-Markov and secular approximations. Suppose that the bath correlation matrix $\gamma(\nu)$ is full rank for every active Bohr-frequency sector $\nu$. Then the HNLS condition~\cite{zhou2018achieving} $H_S \notin \operatorname{span}\{ \mathbb{1}, L_i, L^\dagger_i, L^\dagger_i L_j\}$ is satisfied implies that the HNIS condition $H_S \notin \operatorname{span}\{\mathbb{1}, S_i\}$ is also  satisfied.
\end{proposition}

 To show this, we use the microscopic derivation of the Lindblad master equation. Following Ref.~\cite{manzano2020short}, we can always write $H_{SE} = \sum_{i,\nu} S_i(\nu) \otimes E_i $. This gives rise to two equivalent expressions of the master equation: 
\begin{align}
    &\dv{\rho_S}{t} +i[H_S, \rho_S]  \\
    & =  \sum_{\nu,i, j} \gamma_{ij}(\nu) \left (  S_i(\nu) \rho_S S_j(\nu)^\dagger - \frac{1}{2}\acomm{S_j(\nu)^\dagger S_i(\nu)}{\rho_S} \right) \\
    &=  \sum_{\nu, i}d_{ii}(\nu)\left (  J_i(\nu) \rho_S J_i(\nu)^\dagger - \frac{1}{2}\acomm{J_i(\nu)^\dagger J_i(\nu)}{\rho_S} \right).
\end{align}

Let $V$ be a unitary matrix that diagonalizes the matrix $\gamma(\nu),$ such that $V^\dagger\gamma(\nu) V = \operatorname{diag} [d_{ii}(\nu)],$ then the operators $J_i(\nu)$ are obtained from $S_i(\nu)$ by the corresponding unitary transformation. The Lindblad jump operators can then be written as $L_i(\nu)=\sqrt{d_{ii}(\nu)}J_i(\nu)=\sqrt{d_{ii}(\nu)}\sum_j V_{i j} S_j(\nu)$. Since $\gamma(\nu)$ is full rank on every sector $\nu$, $d_{ii}(\nu)> 0$ for every corresponding mode. Therefore, multiplication by $\sqrt{d_{ii}(\nu)}$ does not remove any operator direction, while $V(\nu)$ is unitary, and hence
$ \operatorname{span}\{S_i(\nu)\}=\operatorname{span}\{J_i(\nu)\}=\operatorname{span}\{L_i(\nu)\}$
for every active $\nu$. Moreover, because $S_i=\sum_\nu S_i(\nu)$, 
\begin{equation}
    \operatorname{span}\{\mathbb{1},S_i\}\subseteq\operatorname{span}\{\mathbb{1},L_i(\nu)\}\subseteq
\operatorname{span}\{\mathbb{1},L_i,L_i^\dagger,L_i^\dagger L_j\}.
\end{equation} Therefore, if the HNLS condition is satisfied,
$H_S\notin\operatorname{span}\{\mathbb{1},L_i,L_i^\dagger,L_i^\dagger L_j\}$,
then necessarily
$H_S\notin\operatorname{span}\{\mathbb{1},S_i\}$,
which is precisely the HNIS condition. 

The full-rank assumption is essential for this implication. If $\gamma(\nu)$ has a zero eigenvalue, a nonzero linear combination $J_i(\nu)$ of the microscopic coupling operators may have vanishing dissipative rate and therefore makes the operators set incomplete. In that case, the above inclusion of operator spaces is not guaranteed. However, Proposition~\ref{thm::Markov} does not imply that QEC/QZD/DD are equal substitutes whenever the noise is Markovian. This is explained by the following Proposition:

\begin{proposition} \label{prop::HNLS}
    There exists system-environment Hamiltonian $H_{\rm tot}$ and corresponding Lindblad jump operators $\{ L_i\}$ under the Born-Markov and secular approximations where the HNLS condition is not satisfied but the HNIS condition in Theorems~\ref{thm::QZD} and \ref{thm::DD} is satisfied.
\end{proposition}

This can be shown by counterexample. Consider $H_{\rm tot} = \omega\ket{1}_s\bra{1}\otimes \openone_E + \openone_S \otimes H_E + \sigma_x \otimes E$. One may verify that the corresponding Lindblad jump operators after the Born-Markov and secular approximations are $\sigma_+ \coloneqq \ketbra{1}{0}$ and $\sigma_- \coloneqq \ketbra{0}{1}$ so the Lindblad span is $\mathrm{span}\{\ketbra{1}{1}, \ketbra{1}{0}, \ketbra{0}{1}, \ketbra{0}{0} \}$. Since this spans the entire operator space, no signal Hamiltonian $H_S$ can satisfy the HNLS condition. However, the QZD/DD condition is satisfied since $H_S = \omega \ketbra{1}{1} \notin \mathrm{span}\{ \openone, \sigma_x\}$.   A detailed discussion of this counterexample, supported by numerical simulation demonstrating HL achievability using QZD/DD, can be found in Method~\ref{App:Counterexample}.

Proposition~\ref{prop::HNLS} implies that QEC can fail in the Markovian regime because errors caused by the Lindblad jump operators $\{L_i \} $ cannot be corrected without also correcting the signal Hamiltonian $H_S$ in certain situations. QZD/DD avoids this by effectively decoupling the system from the environment,  preventing errors caused by $\{L_i \} $ from accumulating in the first place.

The distinction originates from the different levels at which the two conditions are formulated, HNIS characterizes the microscopic interaction before Markovian coarse-graining, whereas HNLS characterizes the resulting fixed Lindblad generator.
Consequently, microscopic control through QZD/DD may still preserve Heisenberg scaling in regimes not captured by the HNLS condition for the corresponding fixed Lindblad description.

\begin{figure*}[t]
    \centering
    \begin{minipage}[t]{0.45\textwidth}
        \includegraphics[width=\textwidth]{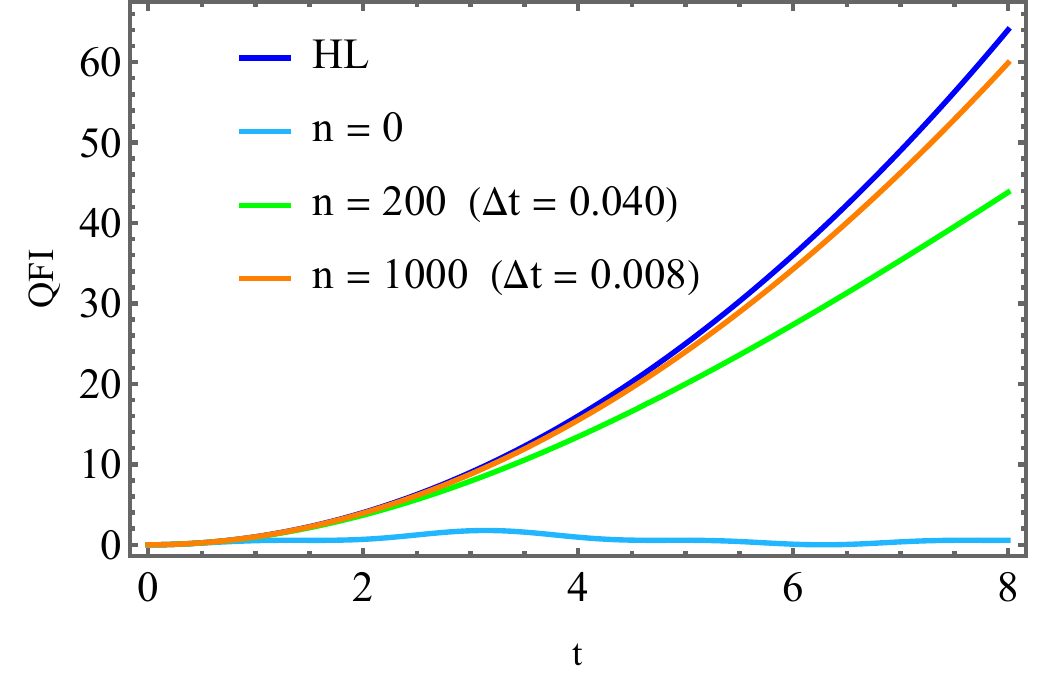}
        \caption{QFI at different projection intervals $\Delta t$ and number of projections $n$. As the projection frequency increases, the QFI-time  curve approaches the Heisenberg limit.}
        \label{fig:first}
    \end{minipage}
    \hfill
    \begin{minipage}[t]{0.45\textwidth}
        \includegraphics[width=\textwidth]{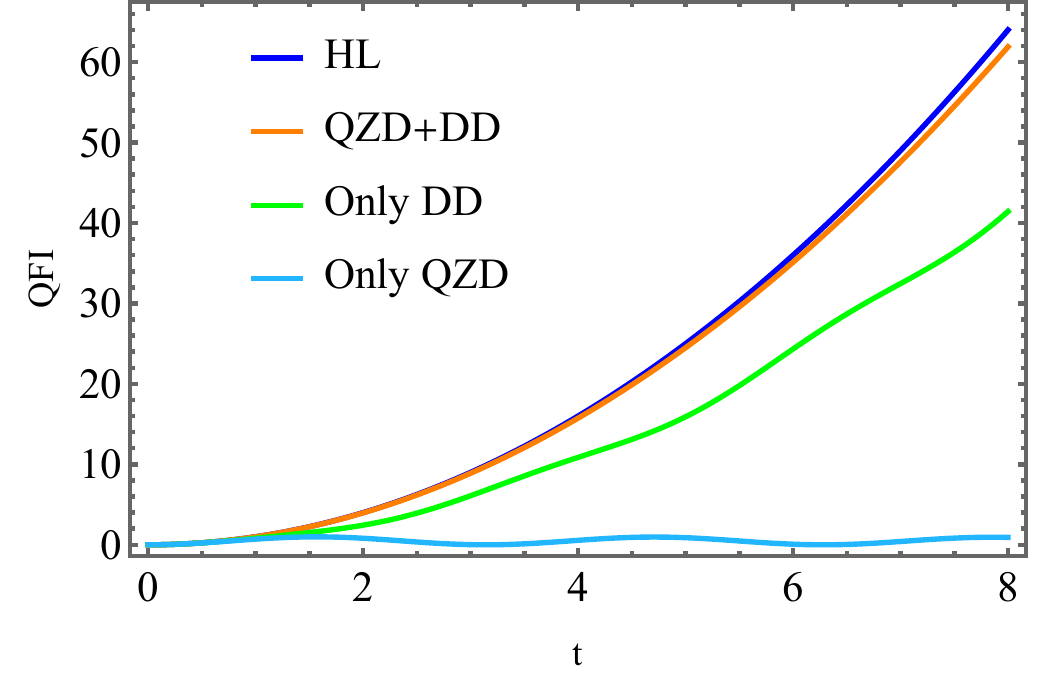}
        \caption{QFI at number of projections $n=1000$, with and without QZD and/or DD. QZD or DD alone were unable to saturate HL, but the combination succeeds.}
        \label{fig:second}
    \end{minipage}
\end{figure*}
\subsection{Hybrid Protocols} \label{sec::hybrid}

In Section~\ref{sec:: DDmain}, it was discussed that even though the HNIS condition for both QZD and DD are identical, the actual physical implementation and technical requirements of both approaches can differ drastically, with QZD requiring access to a noiseless ancilla, and DD requiring arbitrary, instantaneous unitary gates. To lower the technical requirements of each approach, it may help to consider hybrid protocols where QZD and DD are combined in a single protocol.  We find that this hybrid approach is indeed possible under certain circumstances, and the conditions where this can be done is summarized in the following Proposition:

\begin{proposition} 
\label{prop::hybrid}
    Let $\Pi_p$ be a projector acting on the system and ancilla. A hybrid QZD and DD protocol that achieves HL is possible when $\tilde{H}_S  \notin \mathrm{span}\{ \openone, \tilde{S}_i \}$, where $\tilde{H}_S = \Pi_p \openone_a \otimes H_S \Pi_p$,  $\tilde{S}_i \coloneqq \Pi_p \openone_a \otimes S_i \Pi_p$, $H_{SE} = \sum_i S_i \otimes E_i$, and $\{ E_i \}$ is an orthonormal set. 
\end{proposition}

Proposition~\ref{prop::hybrid} is a direct consequence of applying Theorem~\ref{thm::DD} after a QZD projection. To prove the statement, recall that the effective Zeno Hamiltonian after projection is $$H_Z = (\Pi_p \otimes \openone_E) (\openone_a \otimes H_{\rm tot}) (\Pi_p \otimes \openone_E).$$ The effective Zeno system-environment interaction is given by $$\tilde{H}_{SE} = \Pi_p \otimes \openone_E H_{SE} \Pi_p \otimes \openone_E = \sum_i  \tilde{S}_i  \otimes E_i,$$ where $\tilde{S}_i \coloneqq \Pi_p \openone_a \otimes S_i \Pi_p$. The effective Zeno system Hamiltonian after the projection is given by $$ \tilde{H}_S =\Pi_p \openone_a \otimes H_S \Pi_p.$$  The statement in Proposition~\ref{prop::hybrid} is then obtained by applying Theorem~\ref{thm::DD} onto the effective Zeno Hamiltonians $\tilde{H}_{SE}$ and $\tilde{H}_{S}$, resulting in the required condition $\tilde{H}_S \notin \mathrm{span}\{ \openone, \tilde{S}_i \}$ for the hybrid protocol. An example where a hybrid protocol combines two simpler but imperfect QZD and DD protocols to achieve HL is discussed in Section~\ref{sec::examples}.

\subsection{Examples} \label{sec::examples}

Consider the following interaction between a qutrit and the environment: 
\begin{equation}
    H_{\rm tot} = H_S \otimes \mathbb{1}_E + S\otimes E.
\end{equation}
The interpretation of this Hamiltonian is particularly simple. Suppose $E$ is Hermitian and the environment is in a particular state $\ket{\psi}_E ,$ then the system locally ``sees" the Hamiltonian $(H_S+g S)$, where $g= \bra{\psi} E \ket{\psi}$ is some unknown number since the state of the environment is unknown. The noise due to this interaction can therefore be simulated semiclassically by treating $g$ as a random variable drawn from some distribution.

In the first example, choose \begin{equation}
    H_S = \begin{pmatrix}
        \omega/2 & 0 & 0 \\
        0 & -\omega/2 & 0 \\
        0 & 0 & 0 \\ 
    \end{pmatrix} \quad \text{and} \quad  S = \begin{pmatrix}
        0 & 0 & 1 \\
        0 & 0 & 1 \\
        1 & 1 & -1 \\ 
    \end{pmatrix} .
\end{equation} One observes that the QZD condition (Theorem~\ref{thm::QZD}) is satisfied and $H_S \notin \operatorname{span} \{ \mathbb{1}, S\}$ since $H_S$ and $S$ have zero overlap. Choose $\Pi_p = \ketbra{0}{0}+\ketbra{1}{1}$ and initialize the probe in state $(\ket{0}+\ket{1})/\sqrt{2} $. We plot the QFI for various projection frequencies over time in Fig.~\ref{fig:first} and verify that for high projection frequencies, the QFI curve does indeed approach HL. The remaining gap when the number of projections is $n=1000$ can be attributed to leakage due to the finite number of projections, which vanishes in the Zeno limit.

To demonstrate how the hybrid approach considered in Section~\ref{sec::hybrid} may be implemented, consider a second example that has a more complicated Hamiltonian with additional nonzero terms:\begin{equation}
    H_S = \begin{pmatrix}
        \omega/2 & 0 & 0 \\
        0 & -\omega/2 & 0 \\
        0 & 0 & 1 \\ 
    \end{pmatrix} \quad \text{and} \quad  S = \begin{pmatrix}
        1 & 1 & 1 \\
        1 & 1 & 1 \\
        1 & 1 & -1 \\ 
    \end{pmatrix} .
\end{equation}

As in the previous example, the QZD condition $H_S \notin \operatorname{span} \{ \mathbb{1}, S\}$ is still satisfied. However, a noiseless ancilla is required to achieve HL because there are two positive eigenvalues (for detailed explanation of why, see proof of Theorem~\ref{thm::QZD}). Suppose we do not have access to a noiseless ancilla due to experimental constraints, and we perform the same ancilla free projection as in the first example, $\Pi_p = \ketbra{0}{0}+\ketbra{1}{1}.$ This projection alone will not decouple the environment interaction (see Fig.~\ref{fig:second}), but achieving HL is still possible via a hybrid protocol. Indeed, one may verify that $\tilde{H}_S = \Pi_p H_S \Pi_p \notin \mathrm{span} \{ \openone, \tilde{S}\}$, where $\tilde{S} \coloneqq \Pi_p S \Pi_p = \sum_{i,j=0,1} \ketbra{i}{j} $ so Proposition~\ref{prop::hybrid} is satisfied. This implies an ancilla free hybrid protocol can achieve HL by applying the projection $\Pi_p$ in conjunction with dynamical decoupling.

In Fig.~\ref{fig:second} we numerically simulate this hybrid protocol. Taking the total time evolution, we divide it into $n$ intervals where the projective measurement $\Pi_p$ is performed for QZD. Within each of these intervals, we further apply DD pulses that correct rotations about the $X$-axis~\cite{khodjasteh2007performance}. The resulting QFI curves over time are plotted in Fig.~\ref{fig:second}. We see that the hybrid strategy outperforms the individual QZD and DD strategies and saturates the HL QFI curve in the Zeno limit.

\section{Discussion}

In this article, we considered the achievability of HL under noisy conditions using QZD and DD as the primary tools. In our main results, we proved two sets of necessary and sufficient conditions that identify exactly when QZD and DD can succeed (Theorems~\ref{thm::QZD} and \ref{thm::DD}).
Despite the fact that the effective Hamiltonians in the QZD and DD approaches are different, both results are governed by the same HNIS condition.
Our main methodology assumes time independent system-environment interaction, but otherwise imposes no other constraints on the type of noise. These conditions are algebraic, and can be checked directly, assuming that the system-environment Hamiltonian $H_{\text{tot}}$ is known. However, QZD approach requires access to a noiseless ancilla. For DD, we are working within the bang--bang control regime, which assumes instantaneous and arbitrarily strong unitary operations.

Observing that there are some formal similarities between the HNIS condition which determines the achievability of HL for QZD and DD, and the HNLS condition which determines the achievability of HL using QEC for Markovian noise~\cite{zhou2018achieving}, a comparison of the similarities and differences between the HNIS condition and the HNLS condition was made in Section~\ref{sec::HNISvsHNLS}. In Proposition~\ref{thm::Markov}, our findings show that for a large class of Markovian system-environment interactions the HNIS condition is satisfied whenever the HNLS condition is satisfied. Proposition~\ref{prop::HNLS} then shows that there exist instances in the Markovian regime where HNLS fails for the corresponding fixed Lindblad description, while HNIS remains satisfied. This implies QZD/DD can achieve HL in certain instances even when HNLS fails, assuming all the underlying assumptions concerning the QZD and DD approaches are met. QZD/DD and QEC are therefore not completely overlapping approaches to achieving HL, and there will be use cases where one approach is more appropriate than the other.

Comparing QEC in the HNLS framework with QZD/DD in the HNIS framework, we expect QZD to require more frequent measurement intervals than a strategy using QEC. However, QEC requires fast quantum controls in between measurement readouts, which QZD does not need. QEC and QZD methods also assume access to noiseless ancillas, which are expensive resources, while DD only employs fast quantum control but does not require ancillas.  QZD/DD can therefore be less resource intensive alternatives, when the ability to perform QEC is limited by access to fast quantum controls, measurement readouts, or  to noiseless ancillas.  Conversely, we expect QEC to have an advantage when the limiting factor is the speed of the measurement when compared to QZD, or when the limiting factor is the speed of the quantum control when compared to DD. 

Finally, we discuss when a hybrid combination of QZD plus DD methods can achieve HL. Proposition~\ref{prop::hybrid} provides a sufficient condition for hybrid protocols to achieve HL. An explicit example is shown where QZD and DD individually do not saturate HL, but the combination of both methods succeeds. This is confirmed via numerical simulation in Fig.~\ref{fig:second}. It is noteworthy that in this example, the hybrid protocol is ancilla free, and the control pulses used for DD are simple in the sense that they only require rotations about one axis. This suggests hybrid protocols can be useful for designing less resource intensive protocols that can achieve HL. We leave the exploration of this research avenue for future work.

We hope that our work will help the physics community develop practical, noise resilient quantum sensors and protocols.

\section{Methods}
\subsection{Validity of the leading order Magnus approximation under DD}\label{Method:magnus_of_DD}

In the main text, the DD analysis is formulated in terms of the leading-order bang--bang generator. We now show that this approximation becomes asymptotically exact for a cyclic bang--bang implementation in the fast-control limit following the Ref.~\cite{viola2013introduction}. Consider a finite mixed-unitary averaging map
\begin{equation}
\Phi(X)=\sum_{\alpha=1}^{r}p_{\alpha}
V_{\alpha}^{\dagger}XV_{\alpha},
\qquad
p_{\alpha}\geq0,
\qquad
\sum_{\alpha=1}^{r}p_{\alpha}=1.
\label{eq:finite_mixed_unitary_cycle}
\end{equation}
Choose a control cycle of duration $T_c$ and let the control propagator remain equal to $V_{\alpha}$ for the dwell time $\Delta_{\alpha}=p_{\alpha}T_c$. Consecutive propagators are connected by instantaneous system-only pulses, and a final instantaneous pulse closes the cycle, $U_C(T_c)=\mathbb{1}_S$. During the $\alpha$th interval, the toggling-frame Hamiltonian is
\begin{equation}
\widetilde H_{\alpha}=\left(V_{\alpha}^{\dagger}\otimes\mathbb{1}_E\right)H_{\mathrm{tot}}\left(V_{\alpha}\otimes\mathbb{1}_E\right).
\end{equation}

The exact one-cycle propagator is
\begin{equation}
\widetilde U(T_c)=e^{-i\Delta_r\widetilde H_r}\cdots e^{-i\Delta_1\widetilde H_1}=e^{-iT_c\overline H(T_c)},
\end{equation}
where
$\overline H(T_c)=\sum_{\ell=0}^{\infty}\overline H^{(\ell)}$
is the Magnus average Hamiltonian. Its first two terms are
\begin{align}
\overline H^{(0)}
&=
\sum_{\alpha=1}^{r}p_{\alpha}\widetilde H_{\alpha}
\nonumber\\
&=
\Phi(H_S)\otimes\mathbb{1}_E
+\mathbb{1}_S\otimes H_E
+\Phi\otimes\openone _E(H_{SE}),
\label{eq:zeroth_order_cyclic_DD}\\
\overline H^{(1)}
&=
-\frac{iT_c}{2}
\sum_{\alpha>\beta}
p_{\alpha}p_{\beta}
\left[\widetilde H_{\alpha},\widetilde H_{\beta}\right].
\end{align}
Thus, the zeroth-order action of the cyclic block is the mixed-unitary average $\Phi$, while the first correction is proportional to $T_c$.

Assume that the toggling-frame Hamiltonian is bounded and define $\kappa:=\sup_{0\leq s\leq T_c}\|\widetilde H(s)\|_{\infty}$. For sufficiently short control cycles, average Hamiltonian theory~\cite{viola2013introduction} gives $\left\|\overline H^{(\ell)}\right\|_{\infty} = O\!\left[\kappa(\kappa T_c)^{\ell}\right]$.
Hence, in the fast-cycle regime $\kappa T_c\ll1$,
\begin{equation}\label{eq:magnus_remainder}
\left\|\overline H(T_c)-\overline H^{(0)}\right\|_{\infty}=O(\kappa^2T_c).
\end{equation}
Because the block is cyclic, it can be repeated $N$ times. At the
stroboscopic interrogation time $t=NT_c$,
\begin{equation}
\widetilde U(t)=\left[\widetilde U(T_c)\right]^N=e^{-it\overline H(T_c)}.
\end{equation}
The Duhamel formula and Eq.~\eqref{eq:magnus_remainder} give
\begin{equation}
\left\|\widetilde U(t)-e^{-it\overline H^{(0)}}\right\|_{\infty}\leq t\left\|\overline H(T_c)-\overline H^{(0)}\right\|_{\infty}=O(\kappa^2tT_c).
\end{equation}
Therefore, for fixed interrogation time $t$,
\begin{equation}
\lim_{T_c\rightarrow0}\widetilde U(t)=e^{-it\overline H^{(0)}}.
\end{equation}

Therefore, for cyclic bang--bang control with a bounded toggling-frame Hamiltonian and $\kappa T_c\ll1$, all higher-order Magnus contributions vanish in the limit $T_c\rightarrow0$, while the leading-order DD generator remains finite. This establishes the asymptotic validity of the leading-order DD approximation used in the main text for this class of fast cyclic controls. We do not claim that an arbitrary noncyclic or unstructured bang--bang sequence is completely characterized by its first Magnus term.

\subsection{Mixed unitary characterization of DD}\label{Method:Mixed_unitary_for_DD}
We establish a mixed-unitary characterization of the DD conditions used in the proof of Theorem~\ref{thm::DD}. Specifically, we characterize the existence of a set of system-only control unitaries $\{U_C(kT)\}_{k=0}^{n-1}$ that decouples the system--environment interaction while preserving a nontrivial signal Hamiltonian. The result is stated in the following proposition.

\begin{proposition}\label{proposition4}
    For any system-environment interaction $H_S\otimes \mathbb{1}_E + \mathbb{1}_S \otimes H_E + H_{SE}$ where $H_S(\omega) = \omega G$, there exists a collection of control unitary operators $\{ {U_C(kT)} \}_{k=0}^{n-1}$ that satisfies 
    \begin{equation*}
        H^{\text{eff}}_S = \frac{1}{n}\sum_{k=0}^{n-1} U_C^{\dagger}(kT) \, H_S \, U_C(kT) \not\propto \mathbb{1}_S
    \end{equation*} 
    and 
    \begin{align*}
       H_{SE}^{\text{eff}} &= \frac{1}{n}\sum_{k=0}^{n-1}\left( U_C^{\dagger}(kT)\otimes \mathbb{1}_E \right) H_{SE} \left(U_C(kT)\otimes \mathbb{1}_E\right) \\&= c \mathbb{1}_S \otimes J_E
    \end{align*}
    for some constant $c$ and Hermitian operator $J_E$ if and only if there exists mixed unitary channel $\Phi$ such that $\Phi(H_S)$ is nontrivial,
  \begin{equation*}
      \operatorname{Tr} \left\{\left[ \Phi(H_S) - \operatorname{Tr}(H_S) \frac{\mathbb{1}_S}{d} \right] \Phi \left[ \prescript{}{E}{\braket{\psi|H_{SE}|\psi}_E} \right] \right\}= 0
  \end{equation*}  
    and 
    \begin{equation*}
        \braket{i|\Phi\left[ \prescript{}{E}{\braket{\psi|H_{SE}|\psi}_E} \right]| i} = \mathrm{constant},
    \end{equation*}
    for every eigenvector $\ket{i}$ of $\Phi(H_S)$ and state vector $\ket{\psi}_E$.
\end{proposition}

\begin{proof} 
    First, we prove necessity. Assume that $H^{\text{eff}}_S =\frac{1}{n} \sum_{k=0}^{n-1} U_C^{\dagger}(kT) \, H_S \, U_C(kT)\not\propto \mathbb{1}_S$ and $H_{SE}^{\text{eff}} = \frac{1}{n}\sum_{k=0}^{n-1} \left(U_C^{\dagger}(kT)\otimes \mathbb{1}_E\right) \, H_{SE} \, \left(U_C(kT)\otimes \mathbb{1}_E\right) = c \mathbb{1}_S \otimes J_E$ for some constant $c$ and Hermitian operator $J_E$.

    Since the operator $H^{\text{eff}}_S - \operatorname{Tr}(H^{\text{eff}}_S)\frac{\mathbb{1}_S}{d}$ has trace zero, and we assumed $H^{\text{eff}}_S \not\propto \mathbb{1}_S$, $H^{\text{eff}}_S - \operatorname{Tr}(H^{\text{eff}}_S)\frac{\mathbb{1}_S}{d}$ must be nontrivial and  does not have any component along $\mathbb{1}_S$. Observe that $ \prescript{}{E}{\braket{\psi|H_{SE}^{\text{eff}}|\psi}_E}  \propto \mathbb{1}_S,$ which only has a component along $\mathbb{1}_S.$ This implies that 
    \begin{equation}
        \operatorname{Tr} \left\{\left[ H^{\text{eff}}_S - \operatorname{Tr}(H^{\text{eff}}_S) \frac{\mathbb{1}_S}{d} \right] \prescript{}{E}{\braket{\psi|H_{SE}^{\text{eff}}|\psi}_E} \right\} = 0.
    \end{equation}
    Note that the above expression is nontrivial only because $H^{\text{eff}}_S =\frac{1}{n} \sum_{k=0}^{n-1} U_C^{\dagger}(kT) \, H_S \, U_C(kT) \not\propto \mathbb{1}_S$.  

    Define the map $\Phi(\cdot) = \frac{1}{n}\sum_{k=0}^{n-1}U_C^{\dagger}(kT) \,(\cdot)\, U_C(kT)$. Since $\Phi$ is a convex combination of unitary maps, it is a mixed unitary channel by definition and also a CPTP map. As $\Phi$ is trace preserving, we have $\operatorname{Tr}(H^{\text{eff}}_S) =  \operatorname{Tr}(\Phi(H_S)) = \operatorname{Tr}(H_S)$. By further substituting the expressions $H^{\text{eff}}_S = \Phi(H_S)$ and $\prescript{}{E}{\braket{\psi|H_{SE}^{\text{eff}}|\psi}_E} =  \Phi\left[  \prescript{}{E}{\braket{\psi|H_{SE}|\psi}_E}  \right]$, we obtain the expression
    \begin{equation}
        \operatorname{Tr} \left\{\left[ \Phi(H_S) - \operatorname{Tr}(H_S) \frac{\mathbb{1}_S}{d} \right] \Phi \left[ \prescript{}{E}{\braket{\psi|H_{SE}|\psi}_E} \right] \right\}= 0.
    \end{equation}
    
    Let $\left \{ \ket{i} \right \}_{i=0}^{d-1}$ be the eigenbasis of $\Phi(H_S)$. Since $ \Phi \left[ \prescript{}{E}{\braket{\psi|H_{SE}|\psi}_E} \right] \propto \mathbb{1}_S$ for every $\ket{\psi_E}$, we must have $\braket{i | \Phi \left[ \prescript{}{E}{\braket{\psi|H_{SE}|\psi}_E} \right]| i} = \mathrm{constant}$ for every $i = 0,\ldots, d-1$. Therefore, there exists a mixed unitary map $\Phi$ that satisfies the require conditions such that $\Phi(H_S)$ is nontrivial.

    Next, we prove sufficiency Assume that there exists some mixed unitary map $\Phi$ such that $\operatorname{Tr} \left\{\left[ \Phi(H_S) - \operatorname{Tr}(H_S) \frac{\mathbb{1}_S}{d} \right] \Phi \left[ \prescript{}{E}{\braket{\psi|H_{SE}|\psi}_E} \right] \right\}= 0$, $\braket{i | \Phi \left[ \prescript{}{E}{\braket{\psi|H_{SE}|\psi}_E} \right]| i} = \mathrm{constant}$ for every eigenvector $\ket{i}$ of $\Phi(H_S)$ given any state vector $\ket{\psi}_E$ such that $\Phi(H_S)$ is nontrivial (not proportional to the identity).

    Since $\Phi(H_S)$ is nontrivial, the traceless operator $ \Phi(H_S) - \operatorname{Tr}(H_S) \frac{\mathbb{1}_S}{d} $ is also nontrivial. Writing the operator in its eigenbasis $\left \{ \ket{i} \right \}_{i=0}^{d-1}$, we obtain a diagonal representation 
    \begin{equation}
        \Phi(H_S) - \operatorname{Tr}(H_S) \frac{\mathbb{1}_S}{d}  = \begin{pmatrix} \lambda_0 & 0 &\cdots  & 0\\ 
    0 & \lambda_1 & \cdots & 0 \\ 
    \vdots & &\ddots & \vdots \\
    0 & 0 &\cdots& \lambda_{d-1} \end{pmatrix},
    \end{equation}
    where $\lambda_i$, $i=0,\ldots, d-1$ are the eigenvalues of  $ \Phi(H_S) - \operatorname{Tr}(H_S) \frac{\mathbb{1}_S}{d} $.
    
    Consider the operator $\Phi \left[ \prescript{}{E}{\braket{\psi|H_{SE}|\psi}_E} \right]$. Since $\braket{i | \Phi \left[ \prescript{}{E}{\braket{\psi|H_{SE}|\psi}_E} \right]| i} = \mathrm{constant}$ for every $i = 0, \ldots, d-1$, its matrix representation when written in the basis $\left \{ \ket{i} \right \}_{i=0}^{d-1}$ has the form 
    \begin{equation}
        \Phi \left[ \prescript{}{E}{\braket{\psi|H_{SE}|\psi}_E} \right] = \begin{pmatrix} c & * &\cdots  & *\\ 
    * & c & \cdots & * \\ 
    \vdots & &\ddots & \vdots \\
    * & * &\cdots& c \end{pmatrix},
    \end{equation}
    for some constant $c$. In this basis, the leading diagonals of the matrix $\Phi \left[ \prescript{}{E}{\braket{\psi|H_{SE}|\psi}_E} \right]$ is constant.

    Define the pinching channel $\Pi(\cdot) = \sum_{i=0}^{d-1} \ketbra{i} (\cdot) \ketbra{i}$. Pinching channels are defined as mixed, projective maps of the form $\Pi(\cdot) = \sum_i \Pi_i (\cdot) \Pi_i$ that satisfies $\sum_i \Pi_i = \mathbb{1}$ and $\Pi_i^2 = \Pi_i$. It is known that pinching channels are also mixed unitary channels~\cite{Watrous2018}, so for every pinching channel, there exists some distribution $p(x)$ over unitary matrices $U(x)$ such that 
\begin{equation}
    \Pi(\cdot) = \int dx \, p(x) U(x) (\cdot) U(x)^\dagger.
\end{equation}    
One may verify that the pinching channel satisfies $\Pi \circ \Phi (H_S) = \Phi (H_S)$ and $\Pi \circ \Phi \left[ \prescript{}{E}{\braket{\psi|H_{SE}|\psi}_E} \right] \propto \mathbb{1}_S $ given any $\ket{\psi}_E$. The second condition, in particular, implies that $(\Pi \circ \Phi)\otimes \mathbb{1}_E(H_{SE}) = c \mathbb{1}_S \otimes J_E$ for some Hermitian operator $J_E$. 

     Next, observe that since $\Phi$ and $\Pi$ are both mixed unitary channels, their composite $\Phi' = \Pi \circ \Phi$ is also a mixed unitary channel. This means that there exists some probability distribution density $p'(x)$  satisfying $\int dx \, p'(x)  = 1$ such that $\Phi'(\cdot) = \int p'\left(x\right)dx\, U_C^{\dagger}(x) (\cdot) U_C(x).$ We can discretize the probability density by the following substitution 
\begin{equation}
    \int dx \, p'(x) \approx \sum_{j=0}^{m-1} p'(x_j) \Delta x \approx \sum_{j=0}^{m-1} \frac{m_j}{n},
\end{equation}     
where $\frac{m_j}{n} \approx p'(x_j) \Delta x$, and $n, m_j$ are nonnegative integers satisfying $\sum_{j=0}^{m-1} m_j = n$. This approximation is arbitrarily accurate for sufficiently large $n$.

     This allows us to write 
     \begin{eqnarray}
         \begin{aligned}
             \Phi'(\cdot) &=  \int p'\left(x\right)dx\, U_C^{\dagger}(x) (\cdot) U_C(x) \\&
             = \sum_{j=0}^{m-1} \frac{m_j}{n} U_C^{\dagger}(x_j) (\cdot) U_C(x_j)
         \end{aligned}
     \end{eqnarray}
Next, relabel the unitary operators by replacing variable $x_j$ and the summation over $j$ with variable $kT$ and summation over $k$ instead. One can rewrite the composite channel $\Phi' (\cdot)$ as 
\begin{equation}
         \Phi'(\cdot)  = \frac{1}{n}\sum_{k=0}^{n-1} U_C^{\dagger}(kT) (\cdot) U_C(kT).
\end{equation}

Putting it all together, perform the substitution $n \Phi'(\cdot) \rightarrow \sum_{k=0}^{n-1} U^{\dagger}(kT) (\cdot) U_C(kT)$ to get  $H^{\mathrm{eff}}_S = \frac{1}{n}\sum_{k=0}^{n-1} U_C^{\dagger}(kT) H_S U_C(kT) \not\propto \mathbb{1}_S $ since $\Phi'(H_S) \not\propto \mathbb{1}$, and $H_{SE}^{\text{eff}} =\frac{1}{n} \sum_{k=0}^{n-1} \left(U_C^{\dagger}(kT)\otimes \mathbb{1}_E\right) \, H_{SE} \, \left(U_C(kT) \otimes \mathbb{1}_E\right) = c \mathbb{1}_S \otimes J_E$ for some constant $c$ and Hermitian operator $J_E$,  since $\Phi'\otimes \mathbb{1}_E (H_{SE}) \propto \mathbb{1}_S \otimes J_E$ for some Hermitian operator $J_E$. These are the required conditions and complete the proof.

Note that in the proof of necessity, the condition $\braket{i | \Phi [ \prescript{}{E}{\braket{\psi|H_{SE}|\psi}_E} ]| i} = \mathrm{constant}$ does not imply that $\Phi [ \prescript{}{E}{\braket{\psi|H_{SE}|\psi}_E} ] \propto \mathbb{1}$ as off diagonal elements can be nonzero. HL is achieved by further applying a type of mixed unitary channel called a pinching channel~\cite{Watrous2018}. This identifies the pinching channel as particularly important for the achievability of HL. Indeed, pinching channels play an important role in the construction of a channel that proves the sufficiency of Theorem~\ref{thm::DD} (see Method~\ref{Method:HNIS_with_DD}).

     \end{proof}
Proposition~\ref{proposition4} provides a characterization of dynamical decoupling in terms of the existence of mixed unitary channels that satisfy orthogonality conditions. The astute reader may have noticed that the Proposition~\ref{proposition4} presented here has one extra condition over the criterion used in the proof of Theorem~\ref{thm::DD} in the main text. This extra condition can be removed due to the following Proposition:
     \begin{proposition*}
         
 If the condition $
\braket{i | \Phi \left[ \prescript{}{E}{\braket{\psi|H_{SE}|\psi}_E} \right]| i} = \mathrm{constant}$  for every eigenvector $\ket{i}$ of $\Phi(H_S)$ and state vector $\ket{\psi}_E$ is satisfied, then $\operatorname{Tr} \left\{\left[ \Phi(H_S) - \operatorname{Tr}(H_S) \frac{\mathbb{1}_S}{d} \right] \Phi \left[ \prescript{}{E}{\braket{\psi|H_{SE}|\psi}_E} \right] \right\}= 0$   .\end{proposition*}
     
     \begin{proof}
         Let $\Phi(H_S) = \sum_i \lambda_i \ketbra{i}{i}$ where $\lambda_i$, $\ket{i}$ are the eigenvalues and eigenvectors of $\Phi(H_S).$

         This allows us to write $$\Phi(H_S)- \frac{\mathrm{Tr}H_S}{d} \openone_S = \sum_i \lambda'_i \ketbra{i}{i},$$ where $\lambda_i' = \lambda_i - \frac{\mathrm{Tr}H_S}{d}$ and $\sum_i \lambda'_i = 0$.

        Assume that $
        \braket{i | \Phi \left[ \prescript{}{E}{\braket{\psi|H_{SE}|\psi}_E} \right]| i} = \mathrm{constant}.$ Since $\Phi \left[ \prescript{}{E}{\braket{\psi|H_{SE}|\psi}_E} \right]$ is a matrix acting on subsystem $S$, we can express it in terms of the eigenbasis $\{ \ket{i} \}$ and write $$\Phi \left[ \prescript{}{E}{\braket{\psi|H_{SE}|\psi}_E} \right]=\sum_i c_0 \ketbra{i}{i} + \sum_{i \neq j} c_{ij}\ketbra{i}{j}, $$ where $c_0$ is the constant in the assumption. We then have the required expression
        \begin{align}
            \operatorname{Tr} & \left\{\left[ \Phi(H_S) - \operatorname{Tr}(H_S) \frac{\mathbb{1}_S}{d} \right] \Phi \left[ \prescript{}{E}{\braket{\psi|H_{SE}|\psi}_E} \right] \right\} \\
            &=  \mathrm{Tr}\left \{  \left ( \sum_i \lambda_i' \ketbra{i}{i} \right ) \left (  \sum_i c_0 \ketbra{i}{i} + \sum_{i \neq j} c_{ij}\ketbra{i}{j} \right )  \right \} \\
            &= \mathrm{Tr} \left ( \sum_i \lambda_i'c_0 \ketbra{i}{i}\right ) +\mathrm{Tr} \left ( \sum_{i\neq j} \lambda_i'c_{ij} \ketbra{i}{j}\right ) \\
            &= 0,
        \end{align}
     
     where the final line comes from the fact that $\sum_i \lambda'_i = 0$ and $\mathrm{Tr}(\ketbra{i}{j})=0$ when $i\neq j.$
     \end{proof}

\subsection{HNIS Criterion for QZD}\label{Method:HNIS_with_QZD}

We provide a detailed proof of the first main result, Theorem~\ref{thm::QZD}. The Theorem describes a necessary and sufficient condition, so the argument has to show both directions of the statement.

\subsubsection{Sufficiency of the HNIS Criterion}
To prove sufficiency, suppose the HNIS criterion is satisfied so that $$H_S \notin \operatorname{span}\{ \mathbb{1}, S_i \} \coloneqq \mathcal{I}_{SE},$$ where $\mathcal{I}_{SE}$ is the interaction span defined in Section~\ref{sec::QZDmain}. With respect to the Hilbert-Schmidt inner product, we can decompose $H_S$ into its components parallel and orthogonal to $\mathcal{I}_{SE}$, such that $$H_S = H_{||} + H_{\perp}$$ where $H_{||} \in \mathcal{I}_{SE}$, $H_{\perp} \perp \mathcal{I}_{SE}$. Since $H_{\perp} \perp \mathcal{I}_{SE}$ and $\mathbb{1} \in \mathcal{I}_{SE},$ we must have $\operatorname{Tr}(\mathbb{1} H_{\perp}) = \operatorname{Tr}( H_{\perp}) = 0$, so $H_{\perp}$ is traceless. Any traceless, Hermitian matrix can be written as the difference between two positive semidefinite matrices with equal trace, which allows us to write $$H_{\perp} = \frac{1}{2} \operatorname{Tr}( \abs{ H_{\perp} }) (\rho_0 - \rho_1)$$ where we have normalized the aforementioned positive semidefinite matrices $\rho_0,\rho_1$ so that they are normalized density matrices. As we assume access to an ancilla system, there exist purifications $\ket{p_0}, \ket{p_1}$ for each of the density matrices $\rho_0,\rho_1$, where the purifications are state vectors on the combined Hilbert space of the system plus ancilla. Let $\Pi_p = \ketbra{p_0}{p_0}+\ketbra{p_1}{p_1}$ be a projector onto the two dimensional subspace spanned by $\ket{p_0}, \ket{p_1}$.  Recall that the purifications can always be written as \begin{align*}
    \ket{p_i} = \sum_{j}   \sqrt{\rho_i}\ket{j}_S\ket{\psi_{i,j}}_a ,
\end{align*} where $i=0,1$. If the dimension of the ancilla is sufficiently large, we can always ensure that  $\ket{p_0}, \ket{p_1}$ have orthogonal support on the ancilla system, in the sense that $\braket{\psi_{i,j}|\psi_{i',j'}} = \delta_{ii'} \delta_{jj'}$. One may then verify that \begin{align*} \bra{p_j} S_i \otimes \mathbb{1}_a \ket{p_{j'}} = 0,  \quad(\mathrm{when} \; j \neq j'),\end{align*} and \begin{align} \label{eq::projectorIdentity}\bra{p_j} S_i \otimes \mathbb{1}_a \ket{p_{j}} = \sum_{k} \bra{k}\sqrt{\rho} S_i \sqrt{\rho} \ket{k} =c_i.\end{align} Combining these two expressions gives us the following identity
\begin{align*}  \Pi_p S_i \otimes \mathbb{1}_a \Pi_p = c_i \Pi_p,\end{align*} where $\mathbb{1}_a$ acts on the ancilla and $c_i \coloneqq \bra{p_i} (S_i\otimes \mathbb{1}_a) \ket{p_i }$. This implies that $$\Pi_p \otimes \mathbb{1}_E (H_{SE}) \Pi_p \otimes  \mathbb{1}_E = \sum_i c_i \Pi_p \otimes E_i =  \Pi_p \otimes (\sum_i  c_i E_i).$$ The Hamiltonian after the projection is exactly the effective Hamiltonian of the QZD (see Eq.~\ref{eqn::QZD}). We see that if the initial system plus ancilla state is drawn from the subspace spanned by $\ket{p_0}, \ket{p_1},$, then  $\Pi_p$ effectively acts like the identity on this subspace, and the effective Hamiltonian takes the form $\mathbb{1}_{Sa} \otimes J_E$.  This may be interpreted as saying this subspace is effectively decoupled from the environment and thus noise free. 

For the next step, we consider the parallel component $H_\parallel$. By definition, it lies within the interaction span, so we can always write $H_\parallel =a \mathbb{1}+ \sum_k b_k S_k.$ Using the identity from Eq.~\ref{eq::projectorIdentity}, we have $$  \bra{p_j} H_\parallel \otimes \mathbb{1} \ket{p_j} = a+ \sum_k b_k c_k,$$ for $j=0,1$. This implies that \begin{align}  \label{eq::diffIdentity}\bra{p_0} H_\parallel \otimes \mathbb{1} \ket{p_0} -   \bra{p_1} H_\parallel \otimes \mathbb{1} \ket{p_1} = 0 .\end{align}

Finally, we can verify that \begin{align*}
    \bra{p_j} H_S \otimes \mathbb{1} \ket{p_j} = \operatorname{Tr}(\rho_j H_\parallel)+\operatorname{Tr}(\rho_j H_\perp),
\end{align*} where $j=0,1$. Taking the difference gives \begin{align}
    & \bra{p_0} H_S \otimes \mathbb{1} \ket{p_0} -  \bra{p_1} H_S \otimes \mathbb{1} \ket{p_1}\\ & \qquad =    \operatorname{Tr}\left [ (\rho_0 - \rho_1)H_\perp\right ] \\
    & \qquad = \frac{2}{\operatorname{Tr}(\abs{H_\perp})}\operatorname{Tr}\left [ H_\perp^2\right ] \\
    & \qquad \neq 0, \label{eq::notIdentity}
\end{align} where we have used the identities $H_S = H_{||} + H_{\perp}$ and $H_{\perp} = \frac{1}{2} \operatorname{Tr}( \abs{ H_{\perp} }) (\rho_0 - \rho_1)$ to obtain the final expression. Above, the parallel contribution disappears due to Eq.~\ref{eq::diffIdentity}.

Eq.~\ref{eq::notIdentity}  implies that  $\Pi_p H_S \otimes \mathbb{1}_a \Pi_p \not\propto \Pi_{p}$,  which can be interpreted as saying that the Hamiltonian $H_S \otimes \mathbb{1}$ is nontrivial, since it is not proportional to $\Pi_p$ which serves as the identity operator within the two dimensional  subspace.

Therefore, within the subspace spanned by $\ket{p_0}, \ket{p_1}$, the effective Hamiltonian generates nontrivial unitary dynamics on the system and ancilla while being decoupled from the environment. We have therefore found at least one projector $\Pi_p$ where the effective QZD Hamiltonian achieves HL, which  proves the sufficiency of the HNIS criterion.

\subsubsection{Necessity of the HNIS Criterion}
To prove the necessity of the HNIS criterion, assume HL is achievable via QZD. This means there exists projector $\Pi_p$  on the system and ancilla with corresponding effective Hamiltonian $H_Z = \Pi_p \otimes \mathbb{1}_E (\mathbb{1}_a \otimes  H_{\rm tot}) \Pi_p \otimes \mathbb{1}_E$ achieves HL. By the definition of the QZD protocols considered here, the effective QZD Hamiltonian must be decoupled from the environment, so we must satisfy $$\Pi_p  \otimes \mathbb{1}_E (\mathbb{1}\otimes H_{SE}) \Pi_p \otimes \mathbb{1}_E = \Pi_p \otimes H_E'$$ for some $H_E'$.  Using the decomposition $H_{SE} = \sum_i S_i \otimes E_i$, we obtain the expression \begin{align} \label{eq::necessaryIdentity} \sum_i \Pi_p \mathbb{1}_ \otimes S_i \Pi_p \otimes E_i =  \Pi_p \otimes H_E'.\end{align}

Recall that the HNIS criterion requires $\{ E_i \}$ to be a linearly independent set. For any linearly independent set of operators, there always exists a dual set of operators $\{F_i\}$ that is biorthogonal to ${E_i}$, in the sense that it satisfies $\operatorname{Tr}(F_i^\dagger E_j) = \delta_{ij}$. Using this property, we can write \begin{align*}&\operatorname{Tr}_E\left \{(\mathbb{1}_{aS}\otimes F_i^\dagger ) \Pi_p  \otimes \mathbb{1}_E (\mathbb{1}_a\otimes H_{SE}) \Pi_p \otimes \mathbb{1}_E \right\} \\
& \quad= \sum_j \Pi_p \mathbb{1}_a \otimes S_j \Pi_p\operatorname{Tr}(F_i^\dagger E_j) \\
&\quad = \Pi_p \operatorname{Tr}(F_i^\dagger H'_E),\end{align*} where the last equality comes from Eq.~\ref{eq::necessaryIdentity}. Since $\operatorname{Tr}(F_i^\dagger E_j) = \delta_{ij}$, we obtain the expression $$\Pi_p \mathbb{1}_a \otimes S_i \Pi_p = c_i \Pi_p,$$ where $c_i \coloneqq \operatorname{Tr}(F_i^\dagger H'_E).$

If $H_S \in \operatorname{span}\{\mathbb{1}, S_i\}$, then we can always write $H_S = a \mathbb{1}+\sum_k b_k S_k$. This implies \begin{align*}&\Pi_p \mathbb{1}_a \otimes H_S \Pi_p \\
& =a \Pi_p + \sum_k b_k \Pi_p \mathbb{1}_a \otimes S_k \Pi_p  \\
& = (a+\sum_k b_k c_k )\Pi_p \\
&  \propto \Pi_p\end{align*} Since the projected Hamiltonian term $\Pi_p H_{S}\otimes \mathbb{1}_a \Pi_p \propto  \Pi_p$ is the effective QZD Hamiltonian, it acts trivially on the subspace $\Pi_p$ since $\Pi_p$ is effectively the identity operator on this subspace. This implies that decoupling the system from environment using QZD is impossible if $H_S \in \operatorname{span}\{\mathbb{1}, S_i\}$. Therefore, we must conclude that decoupling and achieving HL is only possible when the HNIS criterion is met, i.e., $H_S \notin \operatorname{span}\{\mathbb{1}, S_i\}$, which  proves the necessity of the criterion.

\subsection{HNIS criterion for DD}\label{Method:HNIS_with_DD}
We now prove the second main result in Theorem~\ref{thm::DD} . By Proposition~\ref{proposition4}, the dynamical decoupling problem can be formulated in terms of the existence of a mixed-unitary channel that suppresses the system--environment interaction while preserving a nontrivial signal Hamiltonian. We show below that this is possible if and only if $H_S\notin\operatorname{span}\{\mathbb{1},S_i\}$.

\subsubsection{Necessity of the HNIS Criterion}
Assume that there exists a mixed unitary channel $\Phi$ that satisfies Proposition~\ref{proposition4} such that 
\begin{equation}\label{DD condition}
    \Phi\left(H_S\right) \not \propto \mathbb{1} \text { and }\langle i| \Phi\left[{ }_E\langle\psi| H_{S E}|\psi\rangle_E\right]|i\rangle=\text{constant}.
\end{equation} 

Choose an orthonormal basis for the operator space acting on the environment $\mathcal{H}_E$. Then one can always expand $H_{SE}$ as the linear sum $H_{SE}= \sum_{j}^{} S_j\otimes E_j$. Suppose $H_S \in \operatorname{span}\{ \mathbb{1}, S_i\}$, so we can write $H_S$ as a linear sum of operators $H_S=a I+\sum_k b_k S_k$ for some set of constants $a,b_k$.

Define the pinching channel $\Pi(\cdot)=\sum_i|i\rangle\langle i|(\cdot)|i\rangle\langle i|$ where $|i\rangle$ are the eigenstates of $\Phi\left(H_S\right)$. We may verify the following relation:
$$
 \Phi\left(H_S\right) =\Pi( \Phi\left(H_S\right)) =aI+\sum_k b_k \Pi(\Phi(S_k)).
$$
Due to the assumption in Eq.~\ref{DD condition}, we must have $\Pi \circ \Phi\left({ }_E\langle\psi| H_{S E}|\psi\rangle_E\right) = cI$ given any $|\psi\rangle_E$. This implies that interaction term has been decoupled and $(\Pi \circ \Phi) \otimes I_E\left(H_{S E}\right)=\sum_{k}^{} \Pi  (\Phi(S_k))\otimes E_k=c I_S \otimes J_E$. Since $\{E_k\}$ is an orthonormal set, $\operatorname{Tr}\left(E_j^{\dagger} E_k\right)=\delta_{j k}$. We then perform a partial trace over the environment, yielding the expression
$$
\begin{aligned}
  \Pi  (\Phi(S_j)) =&\operatorname{Tr}_E\left [(I_S\otimes E_j)\sum_{k}^{} \Pi  (\Phi(S_k))\otimes E_k  \right ] \\
 = &c I_S \operatorname{Tr}(E_jJ_E).
\end{aligned}
$$
This implies that $\Phi\left(H_S\right)=(a+\sum_k b_k c_k) I\propto I$, so the effective Hamiltonian acting on $\mathcal{H_S}$ is trivial and HL cannot be achieved when $H_S\in\operatorname{span}\{\mathbb{1},S_i\}$. This contradicts our assumption in Eq.~\ref{DD condition}, hence $H_S \notin \operatorname{span}\left\{I, S_i\right\}$ must be true whenever Proposition~\ref{proposition4} is true. This proves necessity.

\subsubsection{Sufficiency of the HNIS Criterion}
Define the linear map 
$$
\Phi_{K,\beta}(X) \coloneqq \frac{1}{d} \operatorname{Tr}(X) \mathbb{1}+\beta \operatorname{Tr}(K X) K
$$
where operator $K$ is traceless and Hermitian, and $d$ is the dimension of $\mathcal{H}_S$. Since $\operatorname{Tr}(K)=0$,
$ \operatorname{Tr}\left(\Phi_{K,\beta}(X)\right)=\frac{1}{d} \operatorname{Tr}(X) \operatorname{Tr}(\mathbb{1})+\beta \operatorname{Tr}(K X) \operatorname{Tr}(K)= \operatorname{Tr}(X).$ Therefore $\Phi$ is a trace preserving map.

We want to show that for certain choices of $\beta$ and $K$ $\Phi_{K,\beta}(X)$ is a mixed unitary channel.
Consider the spectral decomposition $K=\sum_{m=1}^d \lambda_m|m\rangle\langle m|$, where $\sum_m \lambda_m=0$ since $K$ is traceless. Observe that given any input operator $X$, the channel $\Phi_{K, \beta}(X) =\text{diag}(q_1,q_2,..,q_d)=\text{diag}(\mathbf{q} )$ is a diagonal matrix in the $\{ \ket{m} \}$ eigenbasis such that
$$
\begin{aligned}
q_m= & \langle m |\Phi_{K, \beta}(X) | m  \rangle  \\
 = &\frac{\operatorname{Tr}(X)}{d} +\beta \operatorname{Tr}(K X)\langle m | K| m  \rangle \\
 = &\frac{1}{d} \sum_n p_n+\beta\left(\sum_n \lambda_n p_n\right) \lambda_m\\
= &\sum_n^d\left(\frac{1}{d}+\beta \lambda_m \lambda_n\right) p_n
\end{aligned} 
$$
where $p_n=\left \langle n |X | n  \right \rangle $ are the diagonal elements of the operator $X$ in this basis. By defining the matrix
$
A_{m n}:=\frac{1}{d}+\beta \lambda_m \lambda_n,
$
this leads to the linear relation $q_m=\sum_n A_{m n} p_n$, then we have 
$$
\Phi_{K, \beta}(X) =\text{diag}(\sum_{n}^{} A_{1n}p_n ,..,\sum_{n}^{} A_{dn}p_n)=\text{diag}(A\mathbf{p} ). 
$$
The matrix $A$ has the following special properties:
\begin{equation}
\begin{gathered}
\sum_n A_{m n}=\sum_n\left(\frac{1}{d}+\beta \lambda_m \lambda_n\right)=1+\beta \lambda_m \sum_n \lambda_n=1 . \\
\sum_m A_{m n}=1+\beta \lambda_n \sum_m \lambda_m=1 .
\end{gathered}
\end{equation}
Therefore, the sums of its rows and columns are both equal to 1. If we can guarantee that all matrix elements $A_{m n}$ are nonnegative, then $A_{mn}$ is a doubly stochastic matrix. This can be guaranteed by choosing $\beta$ such that for any given $K$,  we have $|\beta| \leq \frac{1}{d\|K\|^2}$. This inequality comes from the fact that $\left|\lambda_m \lambda_n\right| \leq\|K\|^2$, where $\lVert (\cdot) \rVert$ is the operator norm.

Next, we exploit the fact that any doubly stochastic matrix $A$ can be written as a convex combination of permutation matrices due to the Birkhoff-von Neumann Theorem, i.e.,
\begin{equation}
A = \sum_i w_i P_i, \qquad 
w_i \ge 0, \qquad 
\sum_i w_i = 1,
\end{equation}
where each $P_i$ is a permutation matrix acting on a column vector $\mathbf{p}$. 

This allows us to write $\Phi_{K, \beta}(X) =\text{diag}(A\mathbf{p} )=\text{diag}(\sum_i w_i P_i\mathbf{p})$. Each permutation matrix $P_i$ acting on a column vector corresponds to a unitary permutation operation $U_i$ acting on a diagonal matrix, which allows us to write 
$$
U_i \operatorname{diag}(\mathbf{p})U_i^{\dagger}=\operatorname{diag}\left(P_i\mathbf{p}\right).
$$

Recall that $p_n$ are the diagonal elements of the operator $X$. It therefore suffices to apply a pinching channel
$
\Pi(X)=\sum_k|k\rangle\langle k| X|k\rangle\langle k|
$
to obtain the diagonal matrix $\operatorname{diag}\left(p_1, p_2, \ldots, p_d\right)$.
We then apply a mixed unitary channel corresponding to the convex combination of permutation operators,
$
\Gamma(\cdot)=\sum_i w_i U_i(\cdot) U_i^{\dagger}
$
which gives us the required channel
\begin{equation}\label{K}
\Phi_{K, \beta}(\cdot)=\Gamma \circ \Pi(\cdot)=\frac{\operatorname{Tr}(\cdot)}{d} \mathbb{1}+\beta \operatorname{Tr}(K(\cdot)) K,
\end{equation}
when $|\beta| \leq \frac{1}{d\|K\|^2}$.
Since both $\Gamma(\cdot)$ and the pinching channel $\Pi(\cdot)$ are mixed unitary channels, it follows that $\Phi_{K, \beta}(\cdot) = \Gamma \circ \Pi(\cdot)$ must also be a mixed unitary channel. This is true for any choice of traceless, Hermitian operator $K$.

The final step is to verify that for the correct choice of $K$, the mixed unitary channel $\Phi_{K, \beta}(\cdot)$ satisfies the conditions in Proposition~\ref{proposition4} when $H_S \notin \operatorname{span}\left\{\mathbb{1}, S_i\right\}$. We first decompose
$H_S$ into its parallel and orthogonal components such that $H_S^\parallel \in \operatorname{span}\{ \mathbb{1}, S_i\}$, $ H_S^\perp \perp   \operatorname{span}\{ \mathbb{1}, S_i\} $, and $H_S = H_S^\parallel + H_S^\perp$. 

Choose $K = H^\perp_S / \sqrt{\Tr [(H^\perp_S)^2 ] }$ such that  $$\Phi_{K, \beta}(\cdot)=\frac{\operatorname{Tr}(\cdot)}{d} \mathbb{1} +\frac{\beta}{\Tr [(H^\perp_S)^2 ]} \operatorname{Tr}\left[H_S^{\perp}(\cdot)\right] H_S^{\perp}.$$ Note that $K = H^\perp_S / \sqrt{\Tr [(H^\perp_S)^2 ] } =  G^\perp/\sqrt{\Tr [(G^\perp_S)^2 ]}$ is independent of the parameter $\omega$ when $H_S = \omega G$. This ensures that the quantum channel $\Phi_{K, \beta}$ does not encode any information about the parameter $\omega$ we are estimating. 

It may be verified that $\operatorname{Tr}(H_S^\perp H_S^\parallel) = 0$ and $\operatorname{Tr}(H_S^\perp S_i)= 0$, which implies $$\langle i | \Phi_{K,\beta}  (\prescript{}{E}{\braket{\psi|H_{SE}|\psi}_E}) |i\rangle   \\
      = \frac{\Tr(\sum_j S_j \braket{\psi |E_j |\psi})}{d}$$ 
 which is a constant for every eigenvector $\ket{i}$ and $$ \Phi_{K,\beta} (H_S) = \frac{\Tr(H_S)}{d} \mathbb{1} + \beta H_S^\perp \not\propto \mathbb{1}.$$ This directly satisfies the conditions in Proposition~\ref{proposition4} and proves sufficiency.

\subsubsection{Optimality of $\Phi_{K,\beta}$}

 We now discuss when the mixed unitary channel $\Phi_{K,\beta}$ is the optimal channel. First, we consider what an optimal channel should look like. From Theorem~\ref{thm::DD}, we know that if $H_S \in \operatorname{span}\{ \mathbb{1}, S_i\}$, then HL is not achievable. This is because the operator subspace $\operatorname{span}\{ \mathbb{1}, S_i\}$ generates the noise that acts on the system, and DD averages out the noise. If $H_S \in \operatorname{span}\{ \mathbb{1}, S_i\}$, then it implies that $H_S$ cannot be distinguished from the noise, and will also be averaged out by the DD. Since the encoding of the parameter $\omega$ is performed by $H_S$, HL cannot be achieved if DD succeeds. This implies that only operators are orthogonal to $\operatorname{span}\{ \mathbb{1}, S_i\}$ can survive the DD process.

 For a general $H_S$, we can always decompose it into its parallel and orthogonal components such that $H_S^\parallel \in \operatorname{span}\{ \mathbb{1}, S_i\}$, $ H_S^\perp \perp   \operatorname{span}\{ \mathbb{1}, S_i\} $, and $H_S = H_S^\parallel + H_S^\perp$. The optimal mixed unitary channel $\Phi$ must therefore perfectly preserve the perpendicular component, i.e., $ \Phi(H_S^\perp) = H_S^\perp$. The parallel component $H_S^\parallel$ will necessarily be averaged out since we demand that DD succeeds.

For the mixed unitary channel $\Phi_{K,\beta}$, one may check for optimality by verifying if $\Phi_{K,\beta}(H_S^\perp) = H_S^\perp$ is satisfied. By applying the map to $H_S$, we obtain the expression $$\Phi_{K,\beta}(H_S) = \frac{\operatorname{Tr}(H_S)}{d} \mathbb{1} +\beta H_S^{\perp},$$ which is the effective Hamiltonian acting on the system after DD. Since $H_S = \omega G$, this generates the effective unitary time evolution$$U_\textrm{eff}(t) = e^{-i(\omega\Tr (G)/d) t}e^{-i(\omega \beta G^\perp) t}.$$ Since $\omega \Tr (G)/d$ is a scalar, $e^{-i(\omega\Tr (G)/d) t}$ applies a global phase that has no measurable consequences, so the actual dynamics is governed by the $e^{-i(\omega \beta G^\perp) t}$ term. In order to preserve the dynamics generated by $H_S^\perp = \omega G^\perp,$ we will require $\beta =1$.

Since $|\beta| \leq \frac{1}{d\|K\|^2}$ where $K = H^\perp_S / \sqrt{\Tr [(H^\perp_S)^2 ] }$ this condition can be satisfied when all the eigenvalues of $H^\perp_S$ have equal magnitude $\abs{\lambda}$. In this case, we have $\Vert H_S^\perp \Vert = \abs{\lambda}$,  $\sqrt{\Tr [(H^\perp_S)^2 ] } = \sqrt{\sum_{i=1}^d |\lambda|^2} = \abs{\lambda} \sqrt{d}$ and  $\Vert K \Vert^2 = 1/d $. This allows us to choose $\beta = 1$. The requirement that the spectrum of $H^\perp_S$ has uniform magnitude can be expressed as the matrix condition $(H^\perp_S)^2 \propto  \mathbb{1}$. A consequence is that $\Phi_{K,\beta}$ is optimal for all single qubit systems:

\begin{corollary}
    Let $\mathcal{H_S}$ be the Hilbert space of a single qubit system. Then the mixed unitary channel $$\Phi_{K, \beta}(\cdot)=\frac{\operatorname{Tr}(\cdot)}{d} \mathbb{1}+\beta \operatorname{Tr}(K(\cdot)) K$$ where $\beta = 1$ and $K = H^\perp_S / \sqrt{\Tr [(H^\perp_S)^2 ] }$ is the optimal channel that achieves HL when the condition $H_S \not\in \operatorname{span}\{ \mathbb{1}, S_i\}$ in Theorem~\ref{thm::DD} is satisfied.
\end{corollary}

\begin{proof}
    As discussed above, the channel is optimal when $\beta = 1$, which is possible whenever $(H^\perp_S )^2 \propto \mathbb{1} $. Therefore, we just need to show that $(H^\perp_S )^2\propto \mathbb{1}$ is always satisfied for single qubit systems.

    To do this, use the fact that the operator space is spanned by the Pauli matrices plus the identity, so  $H^\perp_S \in \operatorname{span} \{ \mathbb{1}, \sigma_x, \sigma_y, \sigma_z \}.$ This allows us to write $H^\perp_S = a_0 \mathbb{1} + a_1\sigma_x + a_2 \sigma_y + a_3 \sigma_z$. We also know that $a_0 = 0$ since by definition, $H^\perp_S \perp \operatorname{span}\{ \mathbb{1}, S_i\}, $ so $\Tr(H^\perp_S \mathbb{1} ) = \Tr(H^\perp_S)=0.$

    This gives us $(H^\perp_S)^2 =  \sum_{i,j=1}^3 a_i a_j \sigma_i \sigma_j,$ where $\sigma_1 \coloneqq \sigma_x$, $\sigma_2 \coloneqq \sigma_y$, and $\sigma_3 \coloneqq \sigma_z$. This can be simplified in the following way: \begin{align*}
        (H^\perp_S)^2 &= \sum_{i=1}^3 a_i^2 \sigma_i^2 + \frac{1}{2}\sum_{i \neq j} a_i a_j(\sigma_i \sigma_j + \sigma_j \sigma_i) \\
        &= \sum_{i=1}^3 a_i^2 \mathbb{1} \\
        &\propto \mathbb{1},
    \end{align*} where we used the identities $\sigma_i^2 = \mathbb{1}$ and $\{ \sigma_i ,\sigma_j\} = \sigma_i \sigma_j + \sigma_j \sigma_i = 2\delta_{ij} $. This proves the result. 
\end{proof}

For $n$ qubit systems, the channel is only optimal for special cases, which is stated as the following Corollary:

\begin{corollary}
    Let $\mathcal{H_S}$ be the Hilbert space of an $n$ qubit system. Then the mixed unitary channel $$\Phi_{K, \beta}(\cdot)=\frac{\operatorname{Tr}(\cdot)}{d} \mathbb{1}+\beta \operatorname{Tr}(K(\cdot)) K$$ where $\beta = 1$ and $K = H^\perp_S / \sqrt{\Tr [(H^\perp_S)^2 ] }$ is the optimal channel that achieves HL when the condition $H_S \not\in \operatorname{span}\{ \mathbb{1}, S_i\}$ in Theorem~\ref{thm::DD} is satisfied and $H_S^\perp = \sum_i a_i P_i $ where $\{ P_i \}$ is a set of mutually anti-commuting $n$ qubit Pauli strings.  
\end{corollary}

\begin{proof}
    Recall that $n$ qubit Pauli strings are $n$ fold tensor products of the Pauli matrices plus the identity, i.e. $P_i = \bigotimes^n_{j=1} \sigma_{ij} $, where $\sigma_{ij} \in  \{ \mathbb{1}, \sigma_x, \sigma_y,\sigma_z \} $. Like the single qubit case, these matrices span the operator space of $n$ qubits, so one may always express $H_S^\perp$ as a linear combination of Pauli strings. However, unlike the single qubit case, $n$ qubit Pauli strings can either commute or anti-commute. 
    
    Similar to the previous corollary, we just need to show that $(H_S^\perp)^2 \propto \mathbb{1}$ when $\{ P_i \}$ forms a mutually anti-commuting set. We may verify that  \begin{align*}
        (H^\perp_S)^2 &= \sum_{i=1} a_i^2 P_i^2 + \frac{1}{2}\sum_{i \neq j} a_i a_j(P_i P_j + P_j P_i) \\
        &= \sum_{i} a_i^2 \mathbb{1} \\
        &\propto \mathbb{1},
    \end{align*} where we used the assumption that $\{ P_i, P_j \} = 2\delta_{ij}$.
\end{proof}

\subsection{Example where HNIS is satisfied but HNLS fails} \label{App:Counterexample}

Here, we will consider in greater detail an example of a metrological scenario in the Markovian regime where the HL is not achievable using HNLS based QEC, but is achievable using the HNIS based QZD or DD. This example also proves Proposition 2 in the main text.

Proposition~\ref{prop::HNLS} may be shown by counterexample. Consider the following Hamiltonian:
\begin{align} \label{eq::h_tot}
    H_{\text{tot}} &= H_S \otimes \openone_E + \openone_S \otimes H_E + H_{SE} \\
    &= \omega \begin{pmatrix} 0 & 0 \\ 0 & 1 \end{pmatrix}_S \otimes \openone_E + \openone_S \otimes H_E + \sigma_x \otimes E
\end{align}

We will apply the Born-Markov and the secular approximations to obtain the corresponding Lindblad jump operators $\{ L_i \}$ in the Lindblad master equation. $\{ L_i \}$ describes noise in the Markovian regime. 

In the interaction picture, we have
\begin{align*}
    H_{SE}(t) &= e^{iH_S t} \otimes e^{iH_E t} (\sigma_x \otimes E) e^{-iH_S t} \otimes e^{-iH_E t} \\
    &= \sigma_x(t) \otimes E(t)
\end{align*}

where $\sigma_x(t) \coloneqq e^{iH_S t} \sigma_x e^{-iH_S t}$ and $E(t) \coloneqq e^{iH_E t} E e^{-iH_E t}$.

Since $e^{iH_S t} = \begin{pmatrix} 1 & 0 \\ 0 & e^{i\omega t} \end{pmatrix}$, we can simplify $\sigma_x(t)$ to obtain
\begin{align*}
    \sigma_x(t) &= \begin{pmatrix} 1 & 0 \\ 0 & e^{i\omega t} \end{pmatrix} \begin{pmatrix} 0 & 1 \\ 1 & 0 \end{pmatrix} \begin{pmatrix} 1 & 0 \\ 0 & e^{-i\omega t} \end{pmatrix} \\
    &= \begin{pmatrix} 0 & e^{-i\omega t} \\ e^{i\omega t} & 0 \end{pmatrix} \\
    &= e^{-i\omega t} \sigma_- + e^{i\omega t} \sigma_+.
\end{align*}

We see that $\sigma_x(t)$ is a sum of 2 operators $\sigma_- \coloneqq \ketbra{0}{1}$ and $\sigma_+ \coloneqq \ketbra{1}{0}$ with associated frequencies  $-\omega$ and $+\omega$ respectively. The usual derivation of the Lindblad master equation~\cite{manzano2020short} states that if one expands the interaction term in terms of its frequency components $\nu$:
\[
H_{SE}(t) = \sum_{j,\nu} e^{i\nu t} S_j(\nu) \otimes E_j(t),
\]
then after applying the Born-Markov approximation and the secular approximation the remaining terms describing the time evolution are given by
\begin{align*}
    \frac{d\rho}{dt} &= -i[H_S + H_{LS}, \rho(t)] \, + \\
    &\quad \sum_{j,k,\nu} \gamma_{jk}(\nu) \left( S_j(\nu)\rho(t)S_k^\dagger(\nu) - \frac{1}{2}\{S_k^\dagger(\nu)S_j(\nu), \rho(t)\} \right),
\end{align*}
where the indices $j,k$ run over the number of operators corresponding to each frequency $\nu$.

In this example, there is only one operator associated with each frequency, so $j,k=1$. By substitution into the above expression, we obtain
\begin{align*}
\frac{d\rho}{dt} & = -i[H_S + H_{LS}, \rho(t)] \, + \\
& \quad \gamma(+\omega)\left(\sigma_+ \rho(t) \sigma_+^\dagger - \frac{1}{2}\{\sigma_+^\dagger \sigma_+, \rho(t)\}\right) + \\
&\quad \gamma(-\omega)\left(\sigma_- \rho(t) \sigma_-^\dagger - \frac{1}{2}\{\sigma_-^\dagger \sigma_-, \rho(t)\}\right).
\end{align*}

The above describes how the probe evolves when the interaction with the environment is Markovian, so errors are necessarily generated by the Lindblad jump operators $\{ L_i \}$. Consider the possibility of achieving HL by using QEC. According to the HNLS condition~\cite{zhou2018achieving}, HL is achievable using QEC only when the signal Hamiltonian encoding parameter $\omega$ has a component orthogonal to $\text{span}\{\openone, \sigma_+, \sigma_-, \sigma_+^\dagger, \sigma_-^\dagger, \sigma_+^\dagger \sigma_+, \sigma_-^\dagger \sigma_-, \sigma_+^\dagger \sigma_-, \sigma_-^\dagger \sigma_+\}$. However, one may check that the following terms in this set:
\begin{align*}
\sigma_+ &= \begin{pmatrix} 0 & 0 \\ 1 & 0 \end{pmatrix}, \,\qquad \sigma_- = \begin{pmatrix} 0 & 1 \\ 0 & 0 \end{pmatrix}, \\
\sigma_+^\dagger \sigma_+ &= \begin{pmatrix} 1 & 0 \\ 0 & 0 \end{pmatrix}, \quad \sigma_-^\dagger \sigma_- = \begin{pmatrix} 0 & 0 \\ 0 & 1 \end{pmatrix},
\end{align*}
spans the entire space of linear operators acting on subsystem $S$, so it is not possible to find any signal Hamiltonian that can satisfy the HNLS condition. Since the HNLS condition is both necessary and sufficient for QEC to achieve HL in the Markovian regime, HL cannot be achieved using QEC in this example.

\begin{figure*}[t]
    \centering
    \begin{minipage}[t]{0.45\textwidth}
        \includegraphics[width=\textwidth]{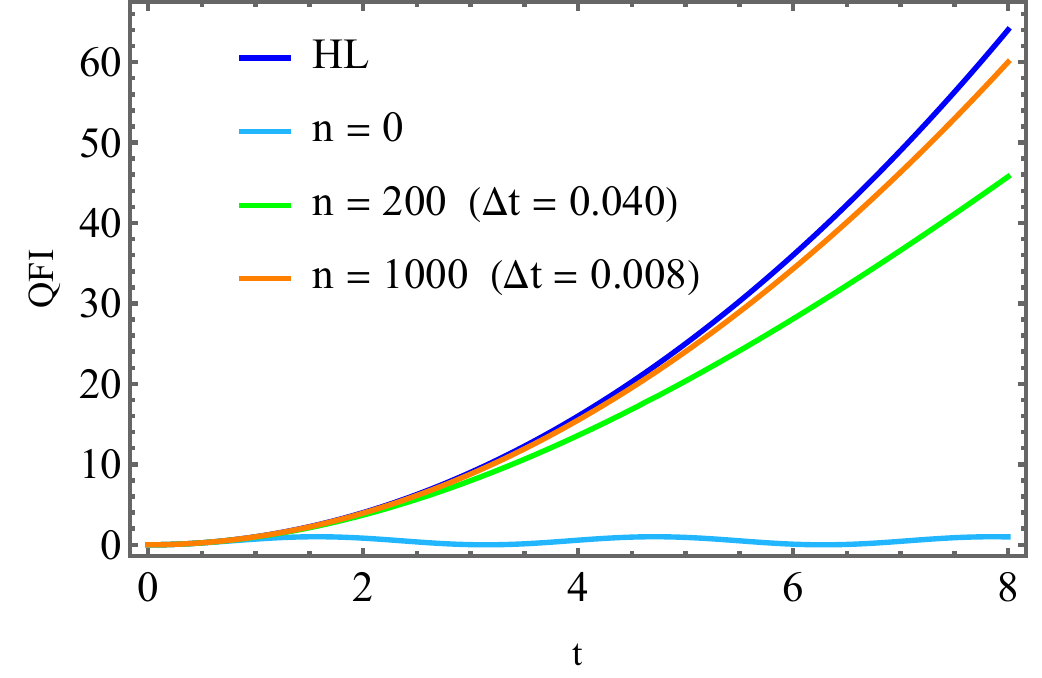}
        \caption{QFI for the counterexample using QZD, at different projection intervals $\Delta t$ and number of projections $n$. As the projection frequency increases, the QFI-time  curve approaches the Heisenberg limit.}
        \label{fig:Appendixfirst}
    \end{minipage}
    \hfill
    \begin{minipage}[t]{0.45\textwidth}
        \includegraphics[width=\textwidth]{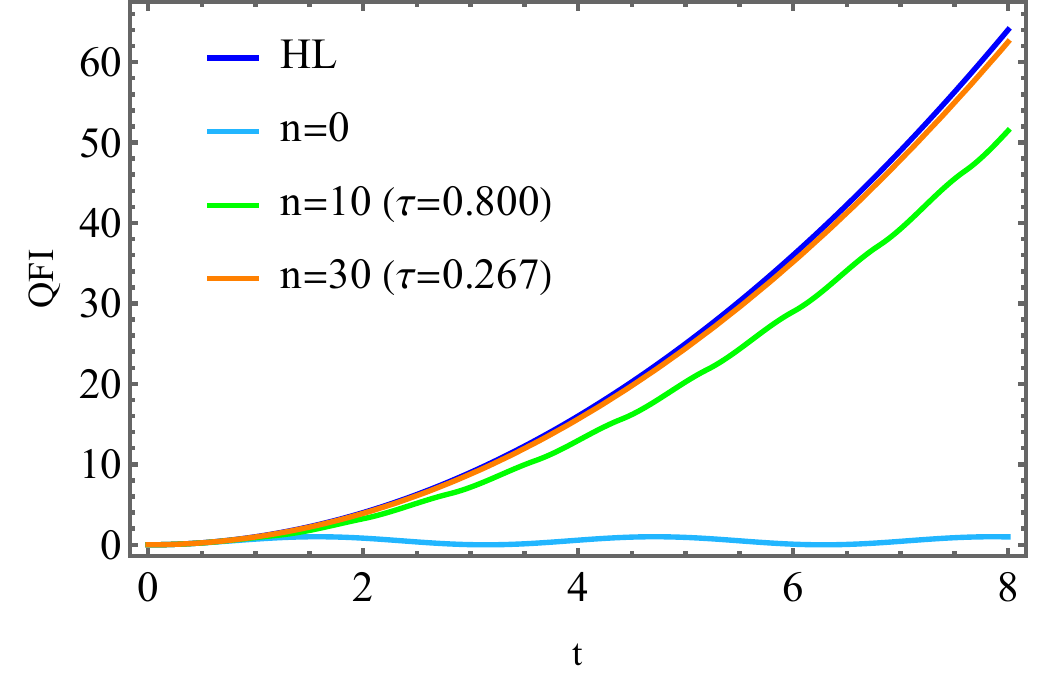}
        \caption{QFI for the counterexample using DD, at different periods $\tau$ and number of cycles $n$. As the number of cycles increases, the QFI-time  curve approaches the Heisenberg limit. }
        \label{fig:Appendixsecond}
    \end{minipage}
\end{figure*}

In the main text, Theorems 1 and 2 provide necessary and sufficient conditions for HL achievability using QZD and DD. In this case, we will use the condition for sufficiency, which states that if $H_S \notin \text{span}(\openone, \sigma_x),$ then HL is achievable using either QZD or DD. Since $H_S = \omega\begin{pmatrix} 0 & 0 \\ 0 & 1 \end{pmatrix} = \frac{\omega}{2}(\openone -\sigma_z)$, it is clear that the condition $H_S \notin \text{span}(\openone, \sigma_x)$ is satisfied so HL is achievable. This implies that HL can be achieved using QZD/DD whereas it is not possible to achieve this using the QEC protocols covered by the HNLS condition. 

To verify this, the following is the explicit QZD protocol that achieves this:

First, append a noiseless ancilla to the system. This results in the extended Hamiltonian:
%\begin{align*}
%H_{\text{tot}}' &= \openone_a \otimes H_{\text{tot}} \\
%&= \omega\openone_a \otimes \begin{pmatrix} 0 & 0 \\ 0 & 1 \end{pmatrix}_S \otimes \openone_E + \openone_a \otimes \openone_S \otimes H_E %+ \openone_a \otimes \sigma_x \otimes E
%\end{align*}

\begin{align*}
H_{\mathrm{tot}}'
&= \openone_a \otimes H_{\mathrm{tot}} \\
&= \omega \openone_a \otimes
\begin{pmatrix}
0 & 0 \\
0 & 1
\end{pmatrix}_S
\otimes \openone_E \\
&\quad + \openone_a \otimes \openone_S \otimes H_E
+ \openone_a \otimes \sigma_x \otimes E .
\end{align*}

Consider the projection of $H_{\text{tot}}'$ onto the $\{|00\rangle, |11\rangle\}$ subspace using the projector $\Pi_p = |00\rangle\langle 00| + |11\rangle\langle 11|$, which acts on the system plus ancilla. The resulting effective Hamiltonian on the subspace is given by
\[
\Pi_p \otimes \openone_E H_{\text{tot}}' \Pi_p \otimes \openone_E = \omega|11\rangle\langle 11| \otimes \openone_E + \Pi_p \otimes H_E
\]
where the interaction term disappears since
\[
\openone_a \otimes \sigma_x = (|00\rangle\langle 01| + |01\rangle\langle 00| + |10\rangle\langle 11| + |11\rangle\langle 10|).
\]

We see that all the operator components acting on the system plus ancilla in the interaction term contain either $|01\rangle$ or $|10\rangle$ vectors, which lie outside of $\text{span}\{|00\rangle, |11\rangle\}$. This implies $\Pi_p (\openone_a \otimes \sigma_x) \Pi_p \otimes E = 0$, thus decoupling the system from the environment. This shows that in the effective Zeno dynamics, the system-environment interactions have been projected away.

Next, prepare the probe state $|\psi\rangle = \frac{1}{\sqrt{2}}(|00\rangle + |11\rangle)$ on the system plus ancilla subsystem. Since this belongs to the $\{|00\rangle, |11\rangle\}$ subspace its effective dynamics is a noiseless unitary evolution described by the Hamiltonian $\omega|11\rangle\langle 11| = \frac{\omega}{2} \left [ \Pi_p - (|00\rangle\langle 00| - |11\rangle\langle 11|) \right ]$.

We see that $\Pi_p$ effectively acts as the identity operator within this subspace, while $(|00\rangle\langle 00| - |11\rangle\langle 11|)$ effectively acts as the Pauli Z operator within this subspace. So the effective dynamics is a noiseless rotation about the Z axis when pictured on the Bloch sphere. This describes the same dynamics as the original signal Hamiltonian $H_S =\frac{\omega}{2}(\openone -\sigma_z)$, which also describes a rotation about the $Z$ axis. In Fig.~\ref{fig:Appendixfirst} we perform a numerical simulation that provides supporting evidence that QZD can retrieve HL in the Zeno limit.

Since the QZD condition is the same as the DD condition (see Theorems 1 and 2 of the manuscript), both QZD and DD can achieve HL. Indeed, we see that in Eq.~\ref{eq::h_tot}, the signal is a rotation about the $Z$ axis while the noise is a random rotation about the $X$ axis. Pulse sequences that are able to decouple this type of error using DD are known~\cite{khodjasteh2007performance}. In Fig.~\ref{fig:Appendixsecond} we perform a numerical simulation that provides supporting evidence that DD can retrieve HL.

The fact that QZD/DD is able to achieve HL while QEC cannot in this example can be attributed to the fact that in QEC the errors are allowed to accumulate over a short period. Since the error is allowed to accumulate, the end result is that when the uncontrolled system--environment dynamics admits a Markovian Lindblad description obtained under the Born--Markov and secular approximations, the errors are characterized by the corresponding Lindblad jump operators $\{L_i\}$. QEC then asks if the error can be corrected after they occur, and in this case the errors cannot be corrected without also correcting the signal Hamiltonian.

In contrast, QZD/DD prevents the accumulation of errors before they occur by decoupling the system from the environment.  In the QZD/DD case, there is effectively no system-environment interaction, the Born–Markov and secular approximations, if applied after introducing the control, would generally yield a different effective generator. This enables the situation described in Proposition~\ref{prop::HNLS} to occur. This example shows that it is sometimes beneficial to prevent errors from occurring using the QZD/DD approach rather than construct a code that corrects errors after they occur.

\begin{acknowledgments}
\noindent\textit{K.C.T. acknowledges support by the National Natural Science Fund of China (Grant No. G0512250610191). } 
\end{acknowledgments}

\bibliography{main.bib}% Produces the bibliography via BibTeX.

\end{document}